\documentclass{article}

\PassOptionsToPackage{numbers, compress}{natbib}

\usepackage[preprint]{neurips_2026}

\usepackage[utf8]{inputenc} 
\usepackage[T1]{fontenc}    
\usepackage[hidelinks,hypertexnames=false]{hyperref} 
\usepackage{url}            
\usepackage{booktabs}       
\usepackage{amsfonts}       
\usepackage{amsmath,amssymb} 
\usepackage{mathtools}
\usepackage{nicefrac}       
\usepackage{microtype}      

\usepackage{enumitem}
\usepackage{graphicx}
\usepackage{multirow}
\usepackage[table]{xcolor}

\usepackage{algorithmic}
\usepackage{algorithm}
\usepackage{amsthm}

\newtheorem{proposition}{Proposition}[section]
\newtheorem{lemma}[proposition]{Lemma}
\newtheorem{corollary}[proposition]{Corollary}
\theoremstyle{remark}

\newcommand{\R}{\mathbb{R}}
\newcommand{\E}{\mathbb{E}}

\title{Generating the Unheard: Phylogeny-Guided Latent Generation for Ancestral Sound Reconstruction}

\author{%
  \textbf{Tianyi Xu}\textnormal{\textsuperscript{1}} \quad
  \textbf{Shrinaath Narasimhan}\textnormal{\textsuperscript{1}} \quad
  \textbf{Evan Gorstein}\textnormal{\textsuperscript{1}} \quad
  \textbf{Santiago Perea}\textnormal{\textsuperscript{1}} \\
  \textbf{Yunyi Shen}\textnormal{\textsuperscript{2}} \quad
  \textbf{Claudia Sol\'is-Lemus}\textnormal{\textsuperscript{1}}\thanks{Corresponding author: Claudia Sol\'is-Lemus. Correspondence: \texttt{txu223@wisc.edu}, \texttt{solislemus@wisc.edu}.} \\
  \textnormal{\textsuperscript{1}}University of Wisconsin--Madison \\
  \textnormal{\textsuperscript{2}}Massachusetts Institute of Technology
}

\begin{document}

\maketitle

\begin{abstract}
What did an ancestral bird species sound like? Existing ancestral state reconstruction methods can infer low-dimensional traits such as morphological characters at internal nodes of a phylogenetic tree, but no one has tried to produce rich perceptual signals such as audio. Some of the challenges include inferred representations that are either too low-dimensional to decode or lie in non-generative feature spaces, so no method to date can produce ancestral audio. 
We introduce the first framework that generates plausible ancestral vocalizations. Our pipeline encodes bird recordings into a VAE latent space, learns a low-dimensional trait projection aligned with phylogenetic distances, performs ancestral inference in this trait space, and recovers decodable latents through an anchored inverse lift before emitting novel waveforms for each ancestral node. Because the entire pipeline stays within a decodable latent space, every internal node receives a genuinely new audio output representing plausible intermediate ancestral sounds unavailable to retrieval-based alternatives. Experiments on two phylogenetically distant bird clades, 21-species Tyrannidae and 19-species Paridae, show that our method is the only approach that simultaneously achieves genuine generation, phylogenetic consistency, and naturalistic audio quality across both datasets.
\end{abstract}

\section{Introduction}

Ancestral state reconstruction (ASR) is a core tool in evolutionary biology~\citep{schluter1997likelihood, pagel1999inferring}, but existing methods are primarily designed for low-dimensional symbolic or continuous traits, such as morphology, size, or discrete character states~\citep{felsenstein1985phylogenies, revell2012phytools}. Many biologically important signals, however, are high-dimensional and perceptual. Bird vocalizations are a prime example: as spectro-temporal signals, their representations require hundreds to thousands of dimensions, yet comparative studies show they carry measurable phylogenetic signal even in vocal learners~\citep{arato2021phylogenetic} and that acoustic divergence correlates with evolutionary distance at macroevolutionary scales~\citep{pearse2018global, derryberry2018ecological}. Reconstructing what an ancestral species sounded like therefore requires going beyond classical trait inference (Figure~\ref{fig:motivation}).
Bird vocalizations also span a fundamental divide in vocal inheritance: suboscine species produce innate, genetically determined calls~\citep{kroodsma1984songs}, while oscine species incorporate culturally learned song elements~\citep{catchpole2008bird}, carrying different degrees of phylogenetic signal~\citep{arato2021phylogenetic}. Moreover, closely related species pairs---such as the cryptic \textit{Empidonax} flycatchers (\textit{E.\ traillii} and \textit{E.\ alnorum}), which are visually near-identical but distinguished primarily by voice---illustrate that vocal divergence can be a primary axis of speciation~\citep{slabbekoorn2002bird, stein1963isolating}, making ancestral sound reconstruction directly relevant to understanding how species form.

\begin{figure}[!tb]
  \centering
  \includegraphics[width=0.85\textwidth]{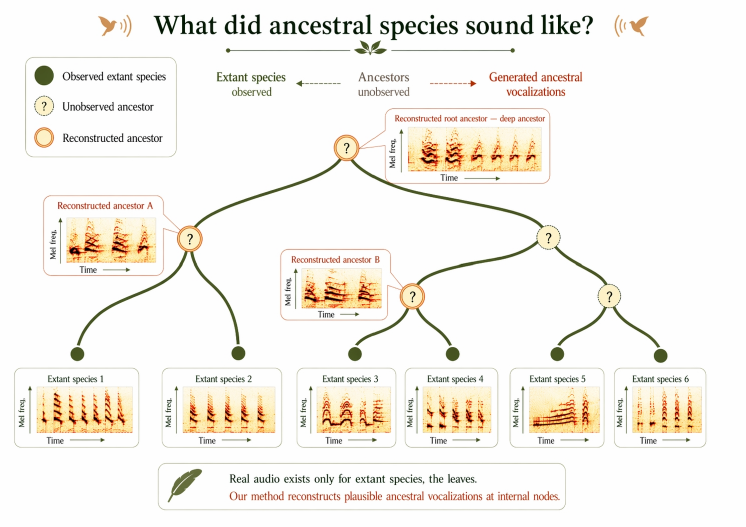}
  \vspace{-15pt}
\caption{\textbf{Motivation.} Extant species at tree leaves have observed vocalizations (mel spectrograms); internal nodes correspond to unobserved ancestors. Our method reconstructs ancestral vocalizations by performing evolutionary inference in a learned latent trait space and decoding back to audio, producing novel spectrograms consistent with descendant species.}
  \vspace{-18pt}
  \label{fig:motivation}
\end{figure}

Existing approaches operate in feature spaces such as acoustic descriptors~\citep{derryberry2018ecological}, classifier embeddings~\citep{kahl2021birdnet}, or learned taxonomic representations, none of which are \emph{decodable} back to audio. No existing method can therefore produce ancestral audio. The closest alternative is to select the most similar extant recording as a proxy, which cannot produce sounds intermediate between descendants or represent vocalizations that no extant species exhibits. The most natural generative alternative, performing Brownian motion ancestral inference~\citep{felsenstein1985phylogenies} directly in a VAE latent space~\citep{kingma2014auto}, also falls short: as we show in Section~\ref{sec:exp}, this naive baseline yields noisy, off-manifold spectrograms (Figure~\ref{fig:qualitative_main}).

We bridge this gap with a four-stage pipeline: (1)~encode extant vocalizations into a pretrained VAE latent space; (2)~learn a tree-metric trait projection aligned with phylogenetic distances; (3)~perform Brownian ancestral inference in this trait space and recover full decodable latents via an anchored inverse lift; (4)~emit novel waveforms through the VAE decoder and a vocoder. The key insight is that the entire pipeline stays within a decodable latent space while preserving phylogenetic signal, so every ancestral node receives a genuinely new audio output.

On two phylogenetically distant clades, 21-species Tyrannidae and 19-species Paridae, our method generates a completely unique waveform for every ancestor while maintaining the strongest phylogenetic plausibility and producing naturalistic spectrograms. Our contributions:
\begin{itemize}[nosep]
  \item The first generative framework for ancestral sound reconstruction, producing novel waveforms at every internal node of a phylogeny rather than selecting existing recordings.
  \item A latent-space decomposition separating phylogenetically informative variation from acoustic texture, enabling evolutionary inference in a compact trait subspace with a closed-form inverse lift back to decodable audio.
  \item Cross-clade generalization across the suboscine--oscine divide: innate vocalizations in Tyrannidae and culturally learned songs in Paridae.
  \item Evaluation protocols and metrics for ancestral audio generation, a task where no ground-truth recordings exist.
\end{itemize}

\section{Related Work}

\noindent\textbf{Ancestral state reconstruction.}\quad
Classical ASR methods infer ancestral values for low-dimensional traits via maximum likelihood~\citep{schluter1997likelihood, pagel1999inferring} or Bayesian estimation~\citep{huelsenbeck2003stochastic, revell2012phytools}, assuming Brownian motion~\citep{felsenstein1985phylogenies} or Ornstein--Uhlenbeck~\citep{hansen1997stabilizing, butler2004phylogenetic} trait evolution. These methods are effective for scalar morphological measurements~\citep{harmon2010early} but do not extend to high-dimensional perceptual signals where $D \gg n$.
Recent work couples generative models with phylogenetic priors: Draupnir~\citep{moreta2022draupnir} uses a tree-structured OU process as a VAE prior for ancestral protein reconstruction, and Phylo-Diffusion~\citep{khurana2024phylodiffusion} conditions diffusion models on tree-of-life hierarchical embeddings to generate images of hypothetical ancestors. Our work differs by operating on audio, using a pretrained generative model and separately learned trait projection, and introducing an anchored inverse lift for the trait-to-latent inverse.

\noindent\textbf{Phylogenetic representation learning.}\quad
Several lines of work learn representations incorporating phylogenetic structure: hyperbolic embeddings for hierarchies~\citep{nickel2017poincare, nickel2018learning, klimovskaia2020poincare}, Siamese Poincar\'e networks for reconstructing evolutionary trees from spectrograms~\citep{carvallo2025siamese}, phylogenetic augmentation for contrastive representation learning~\citep{lu2020evolution}, and phylogenetic PCA~\citep{revell2009phylogenetic}. \citet{mitteroecker2025blomberg} find linear subspaces maximizing Blomberg's $K$~\citep{blomberg2003testing} via eigendecomposition, sharing our goal of phylogenetically informative projections. PhyloNN~\citep{elhamod2023phylonn} learns phylogenetically supervised embeddings via hierarchical classification losses. Our trait learner differs from these approaches in optimizing within a latent space with a patristic distance correlation objective and an orthonormality constraint that enables a closed-form inverse lift. Separately, \citet{li2020latent} decompose VAE latent spaces into attribute-relevant and residual subspaces for image editing. Our row-space and nullspace decomposition serves an analogous role but is learned from phylogenetic distances for ancestral reconstruction.

\noindent\textbf{Neural audio generation and computational bioacoustics.}\quad
Modern audio synthesis combines VAE latent spaces~\citep{kingma2014auto} with neural vocoders~\citep{kong2020hifi, lee2023bigvgan}; diffusion and text-conditioned models~\citep{kong2021diffwave, liu2023audioldm, huang2023makeanaaudio2} further advance quality, but none incorporate evolutionary priors. In bioacoustics, deep learning targets classification, including BirdNET~\citep{kahl2021birdnet}, PANNs~\citep{kong2020panns}, and CLAP~\citep{elizalde2023clap}, as well as representation learning via self-supervised~\citep{hagiwara2023aves} and embedding-based~\citep{ghani2023global} approaches. More broadly, invariance-driven speech representations~\citep{xu2026sita} demonstrate that learned audio embeddings can disentangle biologically relevant factors from nuisance variation. VAE latent spaces capture meaningful vocal structure~\citep{sainburg2020finding}, and comparative studies examine vocal evolution using acoustic features~\citep{derryberry2018ecological, mikula2021global} but remain in non-decodable feature spaces. We bridge audio generation and phylogenetic inference: a pretrained VAE and vocoder serves as a fixed synthesis backend, while our tree-metric trait pipeline produces ancestral vocalizations for the first time.

\section{Method}

\subsection{Setup and Notation}

Let $T=(\mathcal{V}, \mathcal{E})$ be a rooted phylogenetic tree with positive branch lengths $\{w_e\}_{e \in \mathcal{E}}$. The leaf set $\mathcal{L} \subset \mathcal{V}$, $|\mathcal{L}|=n$, indexes extant species; the internal nodes $\mathcal{U} = \mathcal{V} \setminus \mathcal{L}$ are unobserved ancestors. For each $s \in \mathcal{L}$ we observe a collection of audio clips $\{a_{s,i}\}_{i=1}^{n_s}$. Denote by $d_T(s,t) = \sum_{e \in \mathrm{path}(s,t)} w_e$ the patristic distance between $s,t \in \mathcal{V}$ and by $\mathrm{depth}(v) = d_T(\mathrm{root}, v)$ the root-to-$v$ distance for $v \in \mathcal{V}$. The goal is to produce, for each $u \in \mathcal{U}$, a waveform realization consistent with the tree topology, branch lengths, and the extant vocalization distribution.

We assume access to a pretrained audio VAE with encoder $E_{\mathrm{vae}}$, decoder $D_{\mathrm{vae}}$, and a neural vocoder $V$. Each clip is encoded as $x_{s,i} = \mathrm{vec}(E_{\mathrm{vae}}(a_{s,i})) \in \R^D$, and the composition $V \circ D_{\mathrm{vae}}$ maps any latent vector back to a waveform. Our pipeline operates in four stages, \emph{encode}, \emph{evolve}, \emph{lift}, and \emph{emit}, detailed in Sections~\ref{sec:trait}--\ref{sec:lift} and illustrated in Figure~\ref{fig:pipeline}.

\begin{figure}[t]
  \centering
  \includegraphics[width=\textwidth]{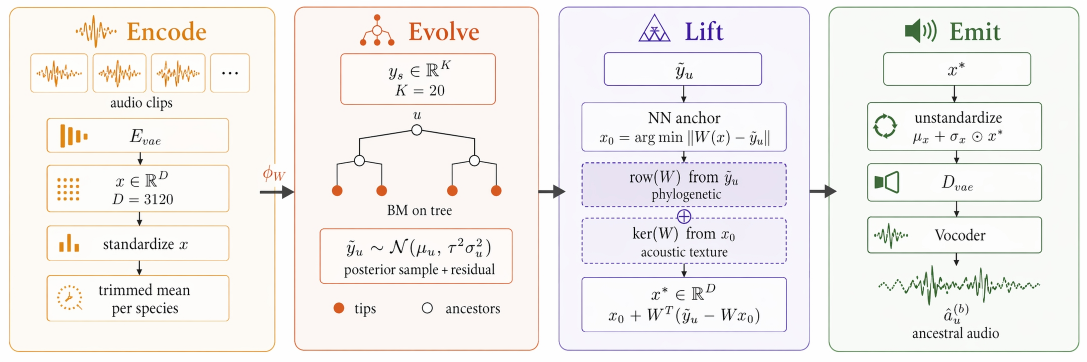}
  \vspace{-20pt}
  \caption{\textbf{Pipeline overview.} Audio clips are encoded into flattened VAE latents $x\in\mathbb{R}^D$ (\emph{Encode}), projected via $W$ into a $K$-dimensional trait space $y\in\mathbb{R}^K$ where Brownian-motion ancestral state reconstruction yields posterior trait vectors $\tilde{y}_u$ for each internal node $u$ (\emph{Evolve}), lifted back to $\mathbb{R}^D$ by combining the phylogenetic component $\mathrm{row}(W)$ from $\tilde{y}_u$ with the acoustic texture $\ker(W)$ from an anchor tip $x_0$ (\emph{Lift}), and decoded through the VAE decoder and vocoder to produce ancestral audio $\hat{a}_u^{(b)}$ (\emph{Emit}).}
  \vspace{-15pt}
  \label{fig:pipeline}
\end{figure}

\subsection{Tree-Metric Trait Space}\label{sec:trait}

Direct phylogenetic inference in $\R^D$ is challenging: although the $n \times n$ BM covariance is cheap to invert, running independent coordinate-wise inference across $D \gg n$ dimensions produces samples that lie off the decodable data manifold, yielding noisy audio (Figure~\ref{fig:qualitative_main}), and generic dimensionality reduction (e.g.\ PCA) does not preserve evolutionary geometry. We instead learn a linear projection $W \in \R^{K \times D}$, $K \ll D$, that defines a \emph{trait representation}
\begin{equation}\label{eq:trait}
  \phi_W(x) \coloneqq Wx.
\end{equation}
The row-orthonormality constraint on $W$ (defined below) ensures that directions in $\R^D$ contribute on a comparable scale after projection.

\smallskip\noindent\textbf{Objective.}\quad
Let $\bar{x}_s = n_s^{-1}\sum_i x_{s,i} \in \R^D$ denote the mean latent for species $s$. We optimize $W$ to maximize the Pearson correlation between pairwise squared trait distances and patristic distances:
\begin{equation}\label{eq:trait-obj}
  \max_{W \in \mathcal{W}}\; \rho\!\Bigl(
    \bigl\{d_T(s,t)\bigr\}_{s<t},\;
    \bigl\{\|\phi_W(\bar{x}_s) - \phi_W(\bar{x}_t)\|_2^2\bigr\}_{s<t}
  \Bigr),
\end{equation}
where $\rho$ denotes the Pearson correlation over all $\binom{n}{2}$ unordered species pairs and $\mathcal{W} = \{W \in \R^{K \times D} : WW^\top = I_K\}$ is the Stiefel manifold $\mathrm{St}(K, D)$. The use of squared Euclidean distances rather than raw distances is motivated by the fact that under Brownian motion on the tree, $\E[\|Y_s - Y_t\|^2]$ is linear in $d_T(s,t)$ when each coordinate evolves as an independent unit-rate Wiener process.

\smallskip\noindent\textbf{Orthonormality via Cholesky reparametrization.}\quad
We constrain $W$ to have orthonormal rows. Direct projection onto the Stiefel manifold after each gradient step can be numerically unstable for $K \ll D$~\citep{wen2013feasible, lezcano2019cheap}. Instead, we introduce an unconstrained parametrization through a free matrix $A \in \R^{K \times D}$ and define
\begin{equation}\label{eq:chol}
  W = L^{-1}A, \qquad L L^\top = A A^\top + \varepsilon I_K,
\end{equation}
where $L \in \R^{K \times K}$ is the (unique) lower-triangular Cholesky factor and $\varepsilon > 0$ is a small regularizer for numerical stability.

\begin{proposition}[Exact row-orthonormality]\label{prop:orth}
Let $A \in \R^{K \times D}$ with $\mathrm{rank}(A) = K$, and let $L$ be the lower-triangular Cholesky factor of $AA^\top + \varepsilon I_K$. In the limit $\varepsilon \to 0$, $W = L^{-1}A$ satisfies $WW^\top = I_K$. For any fixed $\varepsilon > 0$, the deviation satisfies $\|WW^\top - I_K\|_{\mathrm{op}} \le \varepsilon / \sigma_{\min}^2(A)$ (proof in Appendix~\ref{app:lift_proofs}).
\end{proposition}

Gradients of the objective~\eqref{eq:trait-obj} flow through $A$ via the Cholesky decomposition. After each optimizer step, $W$ is recomputed from~\eqref{eq:chol}. The row-orthonormality constraint serves two purposes: (i)~since the Pearson correlation objective is scale-invariant, it removes scale degeneracy, ensuring uniqueness up to rotation and preventing rank collapse; (ii)~it enables the closed-form anchored inverse lift in Section~\ref{sec:lift}, which requires $WW^\top = I_K$. A residual $K \times K$ orthogonal non-identifiability remains, but this is harmless: the squared distances in~\eqref{eq:trait-obj} and the row-space/nullspace decomposition in Section~\ref{sec:lift} are both invariant to such rotations (Appendix~\ref{app:rotation_inv}).

\smallskip\noindent\textbf{Species aggregation.}\quad
After $W$ is learned using simple-mean representatives $\bar{x}_s$, we form robust species-level tip traits for downstream BM inference via a \emph{trimmed mean} in trait space. For each $s \in \mathcal{L}$, compute clip-level traits $z_{s,i} = \phi_W(x_{s,i})$ and the preliminary mean $\bar{z}_s^{(0)} = n_s^{-1}\sum_i z_{s,i}$, then retain the closest fraction $1-\alpha$ of clips:
\begin{equation}\label{eq:agg}
  y_s = \frac{1}{|\mathcal{I}_s^{\mathrm{keep}}|} \sum_{i \in \mathcal{I}_s^{\mathrm{keep}}} z_{s,i}, \qquad
  \mathcal{I}_s^{\mathrm{keep}} = \bigl\{i \in [n_s] : \|z_{s,i} - \bar{z}_s^{(0)}\|_2 \le \delta_s^{(\lceil(1-\alpha)n_s\rceil)}\bigr\},
\end{equation}
where $\delta_s^{(k)}$ denotes the $k$-th smallest value of $\{\|z_{s,i} - \bar{z}_s^{(0)}\|_2\}_{i=1}^{n_s}$ and $\alpha \in [0,1)$ controls robustness to outlier clips. These $\{y_s\}_{s \in \mathcal{L}}$ serve as the observed tip traits for phylogenetic inference.

\subsection{Brownian Ancestral Inference}\label{sec:bm}

We model trait evolution on the tree $T$ as a $K$-dimensional Brownian motion (BM). Under this model, each coordinate evolves independently as a unit-rate Wiener process along the branches of $T$, inducing a joint Gaussian distribution over all node traits with covariance structure determined by shared evolutionary history.

\smallskip\noindent\textbf{Covariance structure.}\quad
For any two nodes $i, j \in \mathcal{V}$, the BM covariance equals the depth of their most recent common ancestor:
\begin{equation}\label{eq:bm-cov}
  C_{ij} = \mathrm{depth}\bigl(\mathrm{MRCA}(i,j)\bigr).
\end{equation}
In particular, $C_{ii} = \mathrm{depth}(i)$ gives the variance at node $i$. Since $C_{ij} = \tfrac{1}{2}(C_{ii} + C_{jj} - d_T(i,j))$ by the tree metric identity, the matrix $C$ is determined by the tree topology and branch lengths.

Denote the tip-tip submatrix of this $|\mathcal{V}| \times |\mathcal{V}|$ covariance by $C_{\mathcal{LL}} \in \R^{n \times n}$. To account for within-species noise and finite-sample effects, we add a nugget term:
\begin{equation}\label{eq:nugget}
  \Sigma_{\mathcal{LL}} = C_{\mathcal{LL}} + \lambda I_n, \qquad \lambda > 0.
\end{equation}

\smallskip\noindent\textbf{Trait standardization.}\quad
Before conditioning, we standardize tip traits coordinate-wise to zero mean and unit variance ($\bar{y}_s = (\hat{\sigma}_y)^{-1} \odot (y_s - \hat{\mu}_y)$), making the unit-rate BM assumption reasonable without estimating a per-coordinate diffusion rate $\sigma^2_k$.

\begin{proposition}[BM conditional posterior]\label{prop:bm}
Under the BM model with covariance~\eqref{eq:bm-cov} and nugget~\eqref{eq:nugget}, the posterior distribution of the trait at internal node $u \in \mathcal{U}$, conditioned on standardized tip traits $\bar{Y}_{\mathcal{L}} \in \R^{n \times K}$, factorizes across coordinates. For each coordinate $j \in \{1,\ldots,K\}$:
\begin{equation}\label{eq:bm-post}
  \bar{Y}_{u,j} \mid \bar{Y}_{\mathcal{L},j} \;\sim\; \mathcal{N}\!\bigl(
    c_u^\top \Sigma_{\mathcal{LL}}^{-1} \bar{Y}_{\mathcal{L},j},\;\;
    C_{uu} - c_u^\top \Sigma_{\mathcal{LL}}^{-1} c_u
  \bigr),
\end{equation}
where $c_u \in \R^n$ is the cross-covariance vector with entries $(c_u)_i = C_{ui}$ for $i \in \mathcal{L}$, and $C_{uu} = \mathrm{depth}(u)$. This follows from standard Gaussian conditioning on the multivariate normal induced by BM on a tree~\citep{felsenstein1985phylogenies, martins1997phylogenies, revell2012phytools}. A self-contained derivation is in Appendix~\ref{app:lift_proofs}.
\end{proposition}

Because trait standardization maps all coordinates to unit variance, the posterior variance $\sigma_u^2 = C_{uu} - c_u^\top \Sigma_{\mathcal{LL}}^{-1} c_u \ge 0$ is a scalar shared across all $K$ coordinates. Write $m_u \in \R^K$ for the posterior mean in standardized space. Denote the unstandardized posterior mean $\hat{y}_u = \hat{\mu}_y + \hat{\sigma}_y \odot m_u$. We draw $B$ posterior samples with a temperature parameter $\tau > 0$ controlling diversity:
\begin{equation}\label{eq:bm-sample}
  y_u^{(b)} = \hat{y}_u + \tau\,\sigma_u\,(\hat{\sigma}_y \odot \epsilon_1), \qquad \epsilon_1 \sim \mathcal{N}(0, I_K),\quad b = 1,\ldots,B.
\end{equation}
Setting $\tau < 1$ concentrates samples around the posterior mean. Because the BM covariance is computed at the raw tree-depth scale while the traits are standardized to unit variance, $\sigma_u^2$ absorbs an implicit diffusion rate. The parameter $\tau$ therefore acts as a joint temperature and rate-calibration factor rather than a pure posterior scaling, and is tuned on downstream quality. The element-wise rescaling by $\hat{\sigma}_y$ maps the noise back to the original trait scale.

\smallskip\noindent\textbf{Residual perturbation.}\quad
Aggregation and posterior averaging suppress within-species variation. To restore clip-level diversity, we perturb each posterior sample with calibrated noise:
\begin{equation}\label{eq:resid}
  \tilde{y}_u^{(b)} = y_u^{(b)} + \gamma\, \hat{\sigma}_{\mathrm{res}} \odot \epsilon_2, \qquad \epsilon_2 \sim \mathcal{N}(0, I_K),
\end{equation}
where $\hat{\sigma}_{\mathrm{res}} \in \R^K$ is the pooled within-species residual standard deviation and $\gamma \ge 0$ controls perturbation strength. The two noise sources are complementary: $\tau$ scales node-dependent phylogenetic uncertainty via $\sigma_u^2$, while $\gamma$ restores within-species acoustic variability uniformly. This preserves the BM posterior mean (Proposition~\ref{prop:resid}, Appendix~\ref{app:resid_prop}).

\subsection{Anchored Inverse Lift}\label{sec:lift}

Since $K \ll D$, the trait map $\phi_W : \R^D \to \R^K$ is many-to-one. Each target trait $y \in \R^K$ determines a $(D{-}K)$-dimensional affine preimage
\begin{equation}\label{eq:preimage}
  \mathcal{A}(y) = \bigl\{x \in \R^D : Wx = y\bigr\} = x_p + \ker(W),
\end{equation}
where $x_p = W^\top y$ is a particular solution (using $WW^\top = I_K$). Recovering a unique decodable latent from $y$ is therefore an underdetermined linear inverse problem. Na\"ive approaches, such as the minimum-norm solution $x_p = W^\top y$, produce latents with zero nullspace component, which lie far from the data manifold and yield poor audio quality. We instead resolve the non-identifiability by \emph{anchoring} the inverse with a nearby observed latent.

\smallskip\noindent\textbf{Anchor selection.}\quad
Let $\mathcal{X}_{\mathrm{train}} = \{x_{s,i}\}_{s \in \mathcal{L},\, i \in [n_s]}$ denote the full set of encoded clips across all species. Given a target (perturbed) ancestral trait $\tilde{y}_u^{(b)}$, we select the training clip whose trait embedding is nearest in $\ell_2$ distance:
\begin{equation}\label{eq:anchor}
  x_0 = \arg\min_{x \in \mathcal{X}_{\mathrm{train}}} \|\phi_W(x) - \tilde{y}_u^{(b)}\|_2.
\end{equation}
This anchor provides a realistic nullspace component: the acoustic texture of a real clip that is phylogenetically close to the target ancestor.

\smallskip\noindent\textbf{Constrained projection.}\quad
We seek the latent $x^\star$ closest to the anchor $x_0$ subject to the trait constraint:
\begin{equation}\label{eq:lift-opt}
  x^\star = \arg\min_{x \in \R^D} \|x - x_0\|_2^2 \qquad \text{subject to} \qquad Wx = \tilde{y}_u^{(b)}.
\end{equation}
This is a linearly constrained least-squares problem with a closed-form solution.

\begin{lemma}[Optimal anchored lift]\label{lem:lift}
Let $W \in \R^{K \times D}$ satisfy $WW^\top = I_K$. For any target trait $y \in \R^K$ and anchor $x_0 \in \R^D$, problem~\eqref{eq:lift-opt} has the unique solution (proof in Appendix~\ref{app:lift_proofs}):
\begin{equation}\label{eq:lift-closed}
  x^\star = x_0 + W^\top\!\bigl(y - Wx_0\bigr).
\end{equation}
\end{lemma}

\begin{corollary}[Exact trait preservation]\label{cor:trait}
The lifted latent satisfies the target trait exactly: $Wx^\star = y$.
\end{corollary}

\begin{corollary}[Nullspace preservation]\label{cor:null}
The correction $x^\star - x_0$ lies entirely in $\mathrm{row}(W)$, so the nullspace component is preserved exactly:
$P_{\ker(W)}\,x^\star = P_{\ker(W)}\,x_0$ where $P_{\ker(W)} = I_D - W^\top W$.
\end{corollary}
The row-orthonormality $WW^\top = I_K$ induces the decomposition $\R^D = \mathrm{row}(W) \oplus \ker(W)$. The anchored lift replaces the \emph{phylogenetic subspace} $\mathrm{row}(W)$ with the BM-inferred target trait while preserving the \emph{acoustic nullspace} $\ker(W)$, which encodes fine-grained spectro-temporal texture from a nearby real clip. This decomposition is by design: the $K$ trait dimensions are precisely those learned to maximize phylogenetic distance correlation (Eq.~\ref{eq:trait-obj}), so they concentrate the evolutionarily informative variation. The remaining $D{-}K$ nullspace dimensions encode acoustic texture, such as timbre, spectral detail, and temporal micro-structure, that is not expected to carry phylogenetic signal and would be underdetermined by the tree alone. Inheriting this component from a nearby extant recording is analogous to ancestral sequence reconstruction, where informative sites are inferred from the phylogeny while neutral positions are filled from the closest reference. Proofs are in Appendix~\ref{app:lift_proofs}.

\smallskip\noindent\textbf{Decoding.}\quad
The lifted latent $x^\star$ is decoded through the pretrained VAE decoder and neural vocoder. For an internal node $u$, the complete ancestral sound generator is
\begin{equation}\label{eq:generator}
  \hat{a}_u^{(b)} = V\!\bigl(D_{\mathrm{vae}}\!\bigl(\Psi(\tilde{y}_u^{(b)})\bigr)\bigr), \qquad b = 1,\ldots,B,
\end{equation}
where $\Psi$ denotes the deterministic anchored-lift operator, comprising anchor selection~\eqref{eq:anchor} and constrained projection~\eqref{eq:lift-closed}. Since $\Psi$, $D_{\mathrm{vae}}$, and $V$ are all deterministic, all stochasticity in the generated output enters through the upstream BM posterior sampling~\eqref{eq:bm-sample} and the residual perturbation~\eqref{eq:resid}. The full pipeline is summarized in Algorithm~\ref{alg:tree_metric_asr} (Appendix~\ref{app:algorithm}).

\section{Experiments}\label{sec:exp}

\noindent \textbf{Datasets.} We evaluate on two phylogenetically and geographically distant bird clades:
\begin{itemize}[nosep,leftmargin=13pt]
  \item \textbf{Dataset~1 (Family Tyrannidae):} 21 New World tyrant flycatcher species spanning 6 genera, namely 10 \textit{Empidonax} spp., 3 \textit{Contopus} spp., 3 \textit{Sayornis} spp., 2 \textit{Tyrannus} spp., 2 \textit{Myiarchus} spp., and 1 \textit{Myiodynastes} sp., with a consensus phylogenetic tree~\citep{jetz2012global} containing 20 internal nodes (Figure~\ref{fig:phylo_tree}). Each species has 47--68 vocalization clips, totaling ${\sim}1{,}100$ clips. Tyrannidae are suboscines whose vocalizations are innate rather than culturally learned~\citep{kroodsma1984songs}, and thus are expected to carry strong phylogenetic signal~\citep{arato2021phylogenetic}.
  \item \textbf{Dataset~2 (Family Paridae):} 19 tit, chickadee, and titmouse species spanning 6 genera, namely 9 \textit{Poecile} spp., 5 \textit{Baeolophus} spp., 2 \textit{Cyanistes} spp., 1 \textit{Parus} sp., 1 \textit{Periparus} sp., and 1 \textit{Lophophanes} sp., with a consensus phylogenetic tree~\citep{jetz2012global} containing 18 internal nodes (Figure~\ref{fig:phylo_tree_xc}). Each species has 50--80 clips, totaling ${\sim}1{,}150$ clips. Paridae are oscines whose songs include culturally learned components~\citep{catchpole2008bird}, which allows us to test whether BM-based ancestral inference remains effective when the phylogenetic-distance assumption is weakened by cultural transmission.
\end{itemize}
All recordings are sourced from Xeno-Canto\footnote{\url{https://xeno-canto.org}} and standardized to 5-second mono 16\,kHz WAV.
The two clades differ in geographic origin, vocal structure, and phylogenetic placement. The same pretrained VAE and BigVGAN vocoder are used for both datasets without retraining; only the tree-metric trait projection $W$ and baseline models are fitted per dataset. Full dataset details, tree visualizations, and preprocessing steps are in Appendices~\ref{app:dataset1} and~\ref{app:dataset2}; details on the pretrained models are provided in Appendix~\ref{app:pretrained}.

\begin{table}[tb!]
\caption{\textbf{Phylogenetic plausibility of reconstructed internal-node sounds.} Methods are either \emph{generative} that produce novel audio: Raw Latent BM, Ours or \emph{retrieval-based} that output copies of training clips: BirdNET BM, Acoustic, PhyloNN, Poincar\'e. Unique Rate $< 1$ indicates retrieval collapse. \textit{Higher is better} for all metrics.}
\label{tab:main_phylo_results}
\centering
\small
\renewcommand{\arraystretch}{1.06}
\resizebox{0.99\linewidth}{!}{
\begin{tabular}{@{}l|ccccc@{}}
\toprule
Method & Tree--Emb.\ Corr.\ $\uparrow$ & Desc.\ Affinity Margin $\uparrow$ & Edge Mono.\ $\uparrow$ & NN-Train Dist.\ $\uparrow$ & Unique Rate $\uparrow$ \\
\midrule
\multicolumn{6}{c}{\textit{Dataset 1: \textrm{Tyrannidae} (21 species)}} \\
\midrule
\multicolumn{6}{l}{\footnotesize\textit{Retrieval-based}} \\
BirdNET Embedding + BM + Decode & 0.4677 & 0.0189 & 0.7333 & 0.2147 & \textbf{1.0000} \\
Classical Acoustic Features & 0.3449 & 0.0289 & 0.8000 & 0.1619 & 0.7500 \\
PhyloNN & 0.2970 & \textbf{0.0845} & 0.8667 & 0.1713 & 0.7000 \\
Poincar\'e & 0.2100 & 0.0063 & 0.9333 & 0.1662 & 0.2500 \\
\midrule
\multicolumn{6}{l}{\footnotesize\textit{Generative}} \\
Raw Latent BM & 0.3205 & 0.0035 & \textbf{1.0000} & \textbf{0.3413} & \textbf{1.0000} \\
\rowcolor{cyan!5}\textbf{Full Method (ours)} & \textbf{0.4859} & 0.0417 & 0.8667 & 0.2350 & \textbf{1.0000} \\
\midrule
\multicolumn{6}{c}{\textit{Dataset 2: \textrm{Paridae} (19 species)}} \\
\midrule
\multicolumn{6}{l}{\footnotesize\textit{Retrieval-based}} \\
BirdNET Embedding + BM + Decode & $-$0.2262 & $-$0.0006 & 0.5833 & 0.1992 & \textbf{1.0000} \\
Classical Acoustic Features & $-$0.3168 & $-$0.0010 & 0.5000 & 0.1726 & 0.5556 \\
PhyloNN & $-$0.2033 & \textbf{0.0491} & 0.7500 & 0.1788 & 0.5556 \\
Poincar\'e & $-$0.0014 & 0.0114 & 0.5833 & 0.1733 & 0.6111 \\
\midrule
\multicolumn{6}{l}{\footnotesize\textit{Generative }} \\
Raw Latent BM & \textbf{0.1848} & 0.0005 & \textbf{0.9167} & \textbf{0.2089} & \textbf{1.0000} \\
\rowcolor{cyan!5}\textbf{Full Method (ours)} & 0.1237 & 0.0023 & \textbf{0.9167} & 0.1742 & \textbf{1.0000} \\
\bottomrule
\end{tabular}
}
\vspace{-15pt}
\end{table}

\noindent \textbf{Baselines.} We compare five baselines spanning two paradigms: \emph{generative} methods decode novel latent vectors to audio, while \emph{retrieval-based} methods output the nearest training clip because their feature spaces are not decodable.
\textbf{Raw Latent BM} (\emph{generative}): BM directly in VAE latent space ($D{=}3{,}120$) without trait projection.
\textbf{BirdNET BM} (\emph{retrieval}): BM in pretrained BirdNET embeddings~\citep{kahl2021birdnet}.
\textbf{Classical Acoustic} (\emph{retrieval}): BM on handcrafted acoustic descriptors.
\textbf{PhyloNN} (\emph{retrieval}): a neural representation trained with phylogenetic supervision~\citep{elhamod2023phylonn}.
\textbf{Poincar\'e} (\emph{retrieval}): hyperbolic embeddings of species relationships~\citep{carvallo2025siamese}.
All retrieval baselines use nearest-exemplar decoding after BM inference, providing an upper bound on retrieval-based performance. Only Raw Latent BM and our method generate novel audio. All methods share the same tree, the same aggregation procedure applied in each method's own embedding space, and the same VAE--vocoder rendering pipeline (Appendix~\ref{app:exp_details}).

\noindent \textbf{Model configuration.} Our full model learns a $K{=}20$ tree-metric trait projection (sensitivity analysis in Appendix~\ref{app:k_sensitivity}) with trimmed aggregation at $\alpha{=}0.3$, BM nugget $10^{-3}$, $B{=}3$ posterior samples at temperature 0.6, residual perturbation $\gamma{=}0.25$, and nearest-trait anchored lifting. All latents are decoded with the same VAE decoder and BigVGAN vocoder. Full details in Appendix~\ref{app:exp_details}.

\noindent \textbf{Evaluation.} Since no ground-truth ancestral audio exists, we evaluate along two axes. \emph{Phylogenetic plausibility} is measured by Tree--Embedding Correlation, Descendant Affinity Margin, Edge Monotonicity, Nearest-Train Distance, and Unique Rate. \emph{Audio quality} is measured by Cycle Similarity, Fr\'echet Audio Distance~\citep{kilgour2019frechet} in BirdNET/PANNs/CLAP embedding spaces, and Silence Rate. Definitions are in Appendix~\ref{app:eval_metrics}.

\section{Results}

\noindent \textbf{Generation vs.\ retrieval.} The central distinction in Table~\ref{tab:main_phylo_results} is between methods that \emph{generate} novel audio and those that \emph{retrieve} existing recordings. Three of the four retrieval baselines collapse to Unique Rates of 0.25--0.75 on Dataset~1 and 0.56--0.61 on Dataset~2, meaning a substantial fraction of their outputs are exact copies of one another. Even BirdNET BM, which retrieves distinct clips for each node, still outputs copies of training recordings rather than novel audio. Our method achieves 100\% Unique Rate on both datasets while generating entirely novel waveforms.

\noindent \textbf{Phylogenetic plausibility within generative methods.} Among generative methods, ours achieves substantially higher TreeCorr on Dataset~1, 0.49 vs.\ 0.32, and higher Descendant Affinity, 0.042 vs.\ 0.004. On Dataset~2, Raw Latent BM has slightly higher TreeCorr, 0.18 vs.\ 0.12, and both tie on Edge Monotonicity at 0.92; however, this comes at the cost of diffuse, noise-like spectrograms with substantially worse FAD scores as shown in Table~\ref{tab:audio_quality}. The lower absolute TreeCorr on Dataset~2 is consistent with culturally learned oscine songs in Paridae carrying weaker phylogenetic signal than innate suboscine vocalizations in Tyrannidae, and with BirdNET being trained predominantly on North American species. Note that BirdNET BM operates in only 1{,}024 dimensions yet performs poorly, especially on Dataset~2, indicating that the bottleneck is not dimensionality alone but the mismatch between a discriminative classifier embedding and the distance structure assumed by BM. On Dataset~2, all non-latent baselines produce \emph{negative} Tree--Embedding Correlations, indicating that only generative methods preserve phylogenetic structure on this clade. Alternative evolutionary models are compared in Appendix~\ref{app:evo_models}; BM provides the best overall balance.

\begin{table}[tb!]
\caption{\textbf{Audio quality of reconstructed ancestral sounds.} Retrieval-based methods achieve low FAD by construction since their outputs are copies of real clips. The meaningful comparison is between generative methods: Raw Latent BM produces degraded audio (high FAD), while our method achieves FAD comparable to retrieval baselines despite generating entirely novel audio.}
\label{tab:audio_quality}
\centering
\small
\renewcommand{\arraystretch}{1.06}
\resizebox{\linewidth}{!}{
\begin{tabular}{@{}ll|ccccc|ccccc@{}}
\toprule
& & \multicolumn{5}{c|}{\textit{Dataset 1: Tyrannidae}} & \multicolumn{5}{c}{\textit{Dataset 2: Paridae}} \\
& Method & CycleSim$\uparrow$ & FAD-BN$\downarrow$ & FAD-PANN$\downarrow$ & FAD-CLAP$\downarrow$ & Silence$\downarrow$
       & CycleSim$\uparrow$ & FAD-BN$\downarrow$ & FAD-PANN$\downarrow$ & FAD-CLAP$\downarrow$ & Silence$\downarrow$ \\
\midrule
\multirow{4}{*}{\rotatebox[origin=c]{90}{\footnotesize\textit{Retrieval}}}
  & BirdNET BM & 0.9985 & \textbf{0.4803} & 0.1681 & 0.6760 & 0.0476 & 0.9981 & \textbf{0.5041} & \textbf{0.1860} & \textbf{0.5317} & 0.0431 \\
  & Acoustic & \textbf{0.9987} & 0.6117 & \textbf{0.1182} & \textbf{0.5983} & 0.0154 & \textbf{0.9984} & 0.6653 & 0.1922 & 0.5739 & 0.0108 \\
  & PhyloNN & \textbf{0.9987} & 0.6165 & 0.1295 & 0.6800 & 0.1000 & \textbf{0.9984} & 0.6111 & 0.2871 & 0.6111 & 0.0615 \\
  & Poincar\'e & 0.9984 & 0.7591 & 0.3388 & 0.7823 & 0.0008 & 0.9979 & 0.5425 & 0.2045 & 0.5532 & 0.0288 \\
\midrule
\multirow{2}{*}{\rotatebox[origin=c]{90}{\footnotesize\textit{Gen.}}}
  & Raw Latent BM & 0.9915 & 0.8587 & 0.6785 & 1.6511 & \textbf{0.0003} & 0.9920 & 0.7575 & 0.6601 & 1.5330 & \textbf{0.0000} \\
  & \cellcolor{cyan!5}\textbf{Ours} & \cellcolor{cyan!5}0.9985 & \cellcolor{cyan!5}0.6326 & \cellcolor{cyan!5}0.1744 & \cellcolor{cyan!5}0.6711 & \cellcolor{cyan!5}0.0079 & \cellcolor{cyan!5}0.9983 & \cellcolor{cyan!5}0.6008 & \cellcolor{cyan!5}0.2028 & \cellcolor{cyan!5}0.5927 & \cellcolor{cyan!5}\textbf{0.0000} \\
\bottomrule
\end{tabular}
}
\vspace{-10pt}
\end{table}

\begin{table}[tb!]
\caption{\textbf{Ablation study on Dataset~1.} Each axis is ablated independently from a base configuration.
\textit{Trait space} and \textit{Lift} ablations use the base config with $\gamma{=}0$ to isolate each component's effect.
\textit{Residual} and \textit{Trim} ablations then vary $\gamma$ and $\alpha$ on top of the best trait and lift choices.
Full method uses $W$ + NN-trait + $\gamma{=}0.25$ + $\alpha{=}0.3$.}
\label{tab:ablation}
\centering
\small
\renewcommand{\arraystretch}{1.06}
\resizebox{\linewidth}{!}{
\begin{tabular}{@{}ll|ccccc|ccccc@{}}
\toprule
& & \multicolumn{5}{c|}{\textit{Phylogenetic Plausibility}} & \multicolumn{5}{c}{\textit{Audio Quality}} \\
Ablation axis & Config
  & TreeCorr $\uparrow$ & DescMarg $\uparrow$ & EdgeMono $\uparrow$
  & NNDist $\uparrow$ & Unique $\uparrow$
  & CycleSim $\uparrow$ & FAD-BN $\downarrow$ & FAD-PANN $\downarrow$
  & FAD-CLAP $\downarrow$ & Silence $\downarrow$ \\
\midrule
\multirow{3}{*}{Trait space}
  & PCA-64            & 0.2450 & 0.0003 & 0.9333 & 0.3717 & \textbf{1.0000} & 0.9933 & 0.9374 & 0.7307 & 1.5836 & 0.0005 \\
  & Random orth.\ $\mathbf{R}$ & 0.1811 & 0.0036 & \textbf{1.0000} & \textbf{0.4413} & \textbf{1.0000} & 0.9936 & 1.0118 & 0.7613 & 1.9633 & \textbf{0.0000} \\
  & Learned $W$ (ours) & \textbf{0.3403} & \textbf{0.0449} & 0.8667 & 0.1877 & \textbf{1.0000} & \textbf{0.9985} & \textbf{0.6584} & \textbf{0.1745} & \textbf{0.6348} & 0.0040 \\
\midrule
\multirow{3}{*}{Lift}
  & Mean              & 0.2230 & $-$0.0005 & \textbf{0.9333} & \textbf{0.3648} & \textbf{1.0000} & 0.9935 & 0.9384 & 0.7283 & 1.5923 & \textbf{0.0000} \\
  & NN-species        & 0.3094 & 0.0079 & \textbf{0.9333} & 0.3541 & \textbf{1.0000} & 0.9934 & 0.8675 & 0.6624 & 1.6282 & 0.3063 \\
  & NN-trait (ours)   & \textbf{0.3403} & \textbf{0.0449} & 0.8667 & 0.1877 & \textbf{1.0000} & \textbf{0.9985} & \textbf{0.6584} & \textbf{0.1745} & \textbf{0.6348} & 0.0040 \\
\midrule
\multirow{3}{*}{Residual $\gamma$}
  & $\gamma = 0$      & 0.3403 & \textbf{0.0449} & \textbf{0.8667} & 0.1877 & \textbf{1.0000} & \textbf{0.9985} & 0.6584 & 0.1745 & \textbf{0.6348} & \textbf{0.0040} \\
  & $\gamma = 0.25$ (ours) & 0.4859 & 0.0417 & \textbf{0.8667} & 0.2350 & \textbf{1.0000} & \textbf{0.9985} & 0.6326 & 0.1744 & 0.6711 & 0.0079 \\
  & $\gamma = 0.5$    & \textbf{0.4943} & 0.0367 & 0.8000 & \textbf{0.2437} & \textbf{1.0000} & \textbf{0.9985} & \textbf{0.5531} & \textbf{0.1576} & 0.6588 & 0.0129 \\
\midrule
\multirow{2}{*}{Trim $\alpha$}
  & $\alpha = 0$        & 0.4810 & 0.0385 & 0.8000 & \textbf{0.2366} & \textbf{1.0000} & \textbf{0.9985} & 0.6343 & \textbf{0.1738} & 0.6833 & \textbf{0.0078} \\
  & $\alpha = 0.3$ (ours) & \textbf{0.4859} & \textbf{0.0417} & \textbf{0.8667} & 0.2350 & \textbf{1.0000} & \textbf{0.9985} & \textbf{0.6326} & 0.1744 & \textbf{0.6711} & 0.0079 \\
\midrule
\rowcolor{cyan!5}\multicolumn{2}{l|}{\textbf{Full method (ours)}}
  & \textbf{0.4859} & \textbf{0.0417} & \textbf{0.8667} & \textbf{0.2350} & \textbf{1.0000} & \textbf{0.9985} & \textbf{0.6326} & \textbf{0.1744} & \textbf{0.6711} & \textbf{0.0079} \\
\bottomrule
\end{tabular}
}
\vspace{-20pt}
\end{table}

\noindent \textbf{Audio quality.}
Table~\ref{tab:audio_quality} reports perceptual quality metrics. Retrieval baselines achieve low FAD by construction since their outputs are verbatim copies of real recordings. Among generative methods, Raw Latent BM produces severely degraded audio with FAD-BN of 0.86 and 0.76 and FAD-CLAP of 1.65 and 1.53 on Datasets~1 and~2, while our method achieves FAD-BN of 0.63 and 0.60 and FAD-CLAP of 0.67 and 0.59, comparable to retrieval baselines despite generating completely novel waveforms that do not exist in the training set.

\noindent \textbf{Ablation study.}
Table~\ref{tab:ablation} ablates each component on Dataset~1. The learned projection $W$ is the most impactful choice, reducing FAD-BN from ${\sim}1.0$ under PCA or random projection to $0.66$. Nearest-trait anchoring is critical: global-mean lift yields FAD-BN $0.94$ and nearest-species anchoring produces 31\% silence. Residual perturbation at $\gamma{=}0.25$ improves TreeCorr from $0.34$ to $0.49$ with negligible quality loss, and trimming at $\alpha{=}0.3$ improves Edge Monotonicity; trait dimension sensitivity is in Appendix~\ref{app:k_sensitivity} and trait space interpretability in Appendix~\ref{app:trait_interpret}.

\begin{figure}[tb!]
\centering
\includegraphics[width=0.98\linewidth]{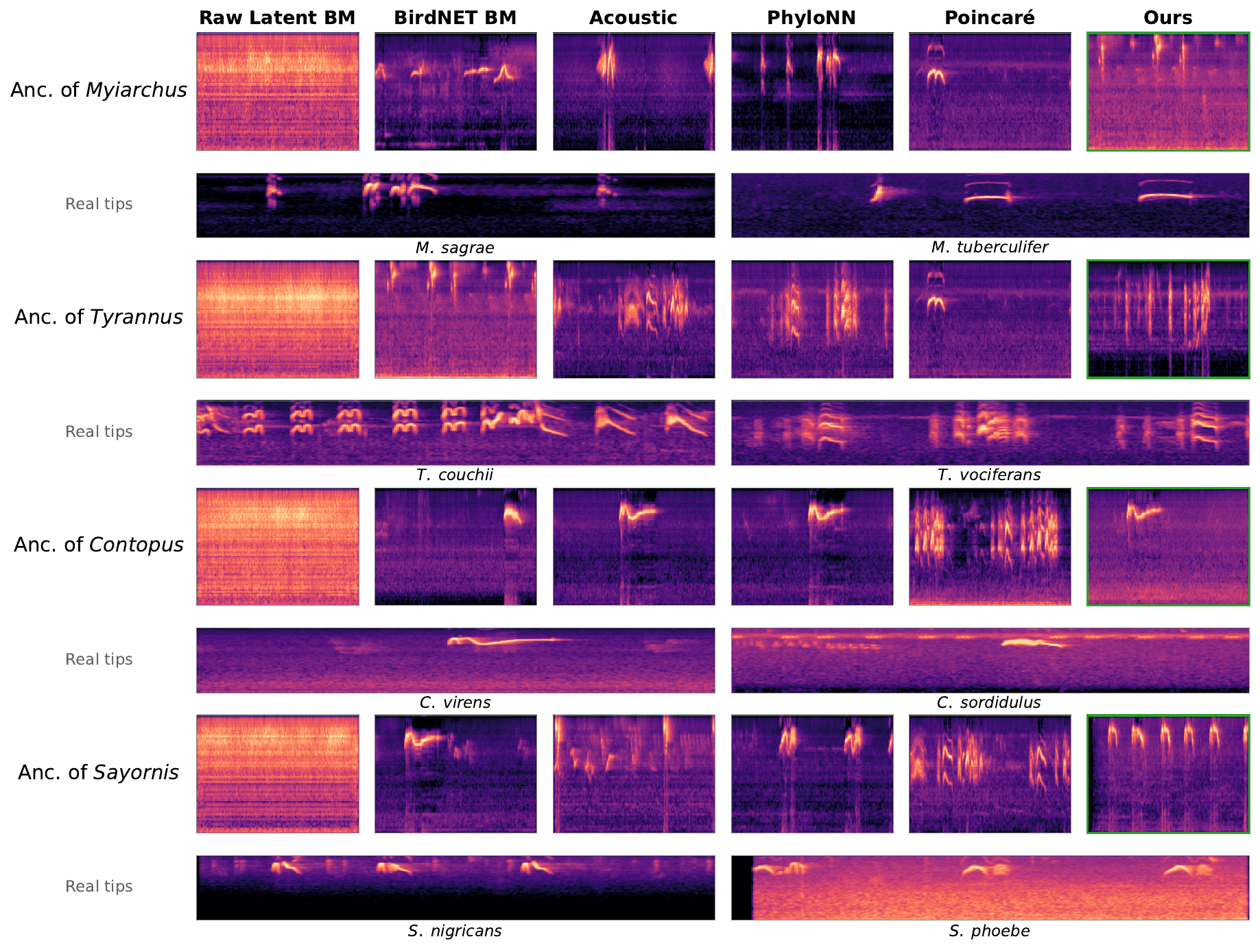}
\vspace{-5pt}
\caption{\textbf{Qualitative comparison of ancestral sound reconstructions.}
Each block shows mel spectrograms generated by six methods for one ancestral node, labeled by genus, with real recordings of the two extant descendant species below for reference.
Retrieval-based baselines such as Acoustic, PhyloNN, Poincar\'e produce copies of training clips and cannot generate novel ancestral sounds.
Our method, highlighted in green, generates novel spectrograms with harmonic and temporal structure that interpolates between the descendant tips.}
\vspace{-18pt}
\label{fig:qualitative_main}
\end{figure}

\noindent \textbf{Qualitative comparison.}
Figure~\ref{fig:qualitative_main} illustrates the generation--retrieval divide. Raw Latent BM produces diffuse, noise-like spectrograms; retrieval baselines produce sharp but duplicated copies of training clips. Our method generates novel spectrograms with clear harmonic structure that inherits recognizable spectro-temporal patterns from descendant species while interpolating between them. The reconstructed MRCA of \textit{T.\ couchii} and \textit{T.\ vociferans} produces a spectrogram whose frequency range and temporal patterning lie between the two descendants rather than copying either one. The reconstructed ancestor of the cryptic \textit{E.\ traillii} and \textit{E.\ alnorum} pair in Figure~\ref{fig:qualitative_appendix_2}, species split taxonomically based on vocal differences, generates an intermediate vocalization blending elements of both descendant calls, suggesting the framework captures the acoustic divergence axis. Extended results are in Appendix~\ref{app:qualitative}.

\section{Conclusion}

We presented the first framework that genuinely generates plausible ancestral vocalizations on a phylogeny, moving beyond retrieval. A tree-metric linear projection provides a compact subspace where Brownian motion matches evolutionary geometry and the anchored inverse lift fills the acoustic nullspace from nearby observed latents with a closed-form solution. Unlike retrieval-based approaches and na\"ive latent-space BM, our method generates unique, naturalistic audio for every ancestral node, with consistent results across two phylogenetically distant clades spanning the suboscine--oscine divide. Beyond methodology, the framework opens avenues for evolutionary biology: testing hypotheses about vocal divergence rates, examining whether reconstructed ancestors of cryptic species pairs exhibit intermediate acoustic properties, and reconstructing vocalizations for recently extinct species. Limitations are discussed in Appendix~\ref{app:limitations}.

\newpage

\bibliographystyle{unsrtnat}
\bibliography{references}

\newpage
\appendix

\section{Limitations}\label{app:limitations}

(1)~\textit{No ground truth.} Ancestral vocalizations are fundamentally unobservable, so evaluation relies on indirect plausibility metrics in BirdNET embedding space; perceptual listening studies could provide complementary validation.
(2)~\textit{Anchor dependence.} The nullspace component of each ancestral latent is inherited from a single anchor clip via Corollary~\ref{cor:null}, so fine-grained acoustic texture is copied from the nearest extant recording rather than inferred from the phylogeny.
(3)~\textit{Scale.} Both datasets contain ${\sim}$20 species; validation on larger phylogenies such as full Passeriformes with more species remains open.

\section{Trait Dimension Sensitivity}\label{app:k_sensitivity}

Table~\ref{tab:k_sensitivity} reports phylogenetic plausibility metrics as a function of the trait dimension~$K$ on Dataset~1.
Descendant Affinity Margin increases with $K$, peaking at $K{=}16$, indicating that higher-dimensional projections better separate descendant from non-descendant affinities.
Tree--Embedding Correlation is highest at $K{=}4$--$8$ and decreases for larger $K$, reflecting a bias--variance trade-off: more trait dimensions allow finer-grained distance matching but also admit more noise from the optimization landscape.
At $K{=}21 = n$, the maximum rank, Tree--Embedding Correlation drops sharply to 0.392, consistent with overfitting when no degrees of freedom remain for regularization.
We select $K{=}20$ as it provides near-maximal expressiveness while retaining regularization capacity, achieving strong Edge Monotonicity of 0.867 and balanced performance across all metrics.
While $K{=}8$ yields the highest TreeCorr of 0.563, its Edge Monotonicity drops to 0.533, the lowest among all settings, indicating that local parent--child acoustic ordering breaks down. Conversely, $K{=}4$ achieves high Edge Monotonicity of 0.933 but retains too few dimensions to capture fine-grained phylogenetic variation, as reflected in its lower DescMarg. $K{=}20$ balances these trade-offs: it maintains strong Edge Monotonicity of 0.867 while encoding enough phylogenetic information for robust ancestral inference across all metrics.

\begin{table}[htb!]
\centering
\caption{\textbf{Trait dimension $K$ sensitivity on Dataset~1.}
All metrics use the full pipeline with learned $W \in \R^{K \times 3120}$.
$\uparrow$: higher is better; $\downarrow$: lower is better. Best per column in \textbf{bold}.}
\label{tab:k_sensitivity}
\small
\begin{tabular}{r ccccc}
\toprule
$K$ & TreeCorr$\uparrow$ & DescMarg$\uparrow$ & EdgeMono$\uparrow$ & NNDist$\uparrow$ & Unique$\uparrow$ \\
\midrule
2   & 0.518          & 0.034          & 0.800          & 0.222          & 1.000 \\
4   & 0.556          & 0.040          & \textbf{0.933} & 0.226          & 1.000 \\
8   & \textbf{0.563} & 0.046          & 0.533          & 0.227          & 1.000 \\
16  & 0.498          & \textbf{0.053} & 0.800          & 0.224          & 1.000 \\
\textbf{20} & 0.486   & 0.042          & 0.867          & 0.235          & 1.000 \\
21  & 0.392          & 0.046          & \textbf{0.933} & \textbf{0.245} & 1.000 \\
\bottomrule
\end{tabular}
\end{table}

\section{Trait Space Interpretability}\label{app:trait_interpret}

\begin{figure}[htb!]
\centering
\includegraphics[width=\linewidth]{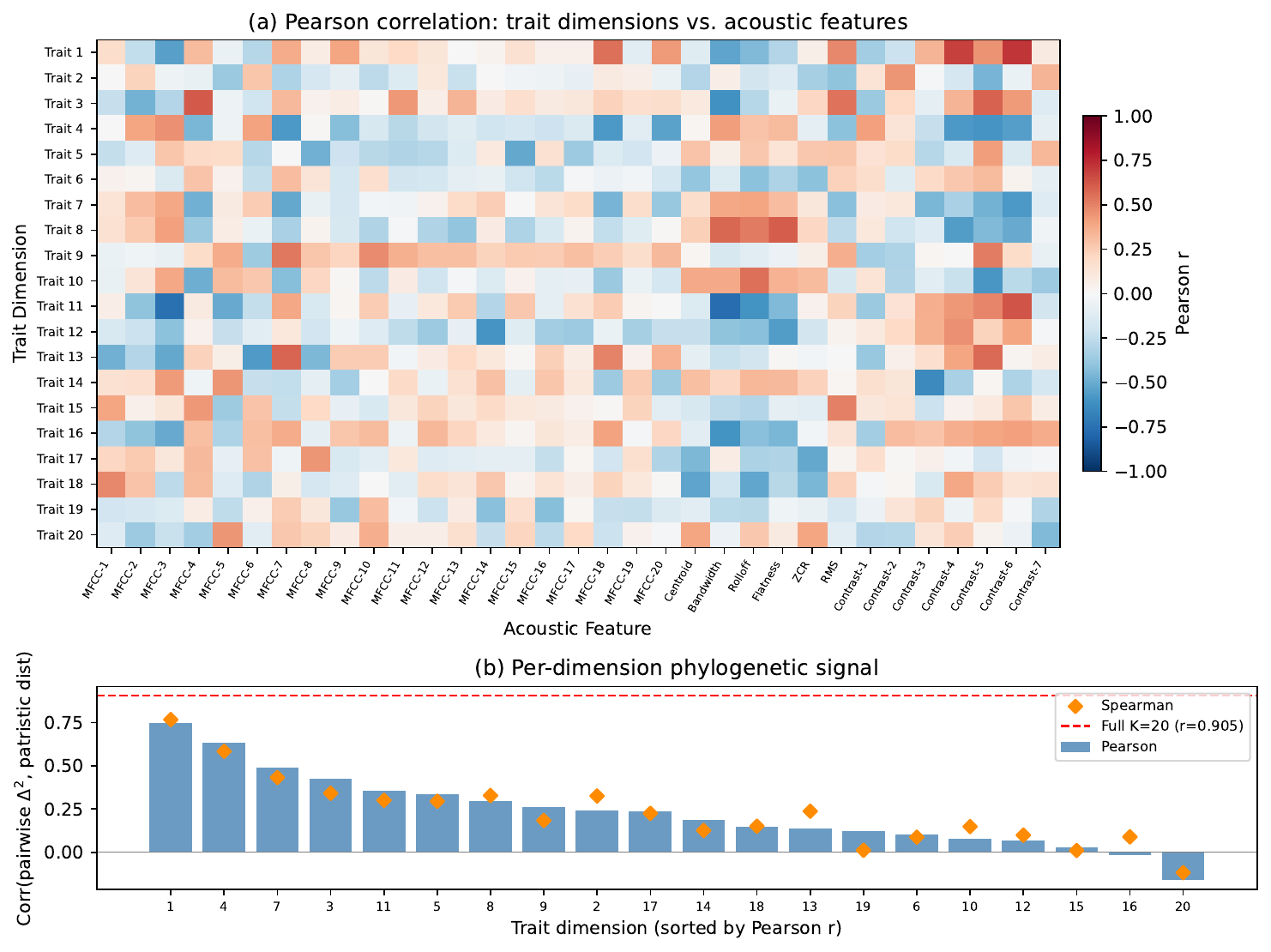}
\caption{\textbf{Trait space interpretability on Dataset~1.}
\textbf{(a)}~Pearson correlation between each of the $K{=}20$ trait dimensions and 33 acoustic features, computed at the species level. Different dimensions capture distinct acoustic properties: spectral contrast, bandwidth, flatness, and mel-frequency cepstral coefficients (MFCCs), which summarize the spectral envelope of a sound.
\textbf{(b)}~Per-dimension phylogenetic signal, measured as the correlation between single-dimension pairwise squared distances and patristic distances. Dimensions are sorted by decreasing signal strength; the dashed line shows the full 20-dimensional correlation ($r{=}0.905$).}
\label{fig:trait_interpret}
\end{figure}

Figure~\ref{fig:trait_interpret} analyzes what the $K{=}20$ learned trait dimensions capture on Dataset~1.

\smallskip\noindent\textbf{Acoustic feature correspondence} (Figure~\ref{fig:trait_interpret}a).\quad
We compute species-level Pearson correlations between each trait dimension and 33 classical acoustic descriptors: 20 MFCCs, spectral centroid, bandwidth, rolloff, flatness, zero-crossing rate, RMS energy, and 7-band spectral contrast.
Different trait dimensions systematically correspond to distinct acoustic properties.
Dim~11 correlates most strongly with spectral bandwidth at $r{=}{-}0.76$, capturing vocal frequency spread.
Dim~1 aligns with high-frequency spectral contrast at $r{=}{+}0.71$, reflecting harmonic structure.
Dim~14 captures mid-frequency spectral contrast at $r{=}{-}0.64$, and Dim~8 tracks spectral flatness at $r{=}{+}0.60$, distinguishing tonal from noisy vocalizations.
The learned trait space thus recovers known acoustic dimensions relevant to comparative bioacoustics~\citep{derryberry2018ecological} without being explicitly trained on them.

\smallskip\noindent\textbf{Per-dimension phylogenetic signal} (Figure~\ref{fig:trait_interpret}b).\quad
We measure each trait dimension's individual phylogenetic signal as the Pearson correlation between its pairwise squared distances and patristic distances.
The dimensions exhibit a clear hierarchy: Dim~1 alone achieves $r{=}0.75$, capturing the majority of the full 20-dimensional signal at $r{=}0.905$.
The top four dimensions---1, 4, 7, and 3---account for most of the phylogenetic structure, while 15 of 20 dimensions maintain positive signal above $r{=}0.1$, confirming that the dimensions are complementary rather than redundant.
Notably, the dimensions with the strongest phylogenetic signal are those most correlated with spectral contrast features, consistent with findings that harmonic structure carries phylogenetic information in bird vocalizations~\citep{arato2021phylogenetic}.

\smallskip\noindent\textbf{Subspace decomposition validation} (Figure~\ref{fig:subspace_decomp}).\quad
Section~\ref{sec:lift} assumes that $\mathrm{row}(W)$ concentrates phylogenetically informative variation while $\ker(W)$ encodes non-phylogenetic acoustic texture.
We verify this empirically by projecting species-mean latents onto each subspace and correlating pairwise squared distances with patristic distances.
The $K{=}20$-dimensional $\mathrm{row}(W)$ achieves $r{=}0.905$, while the remaining $D{-}K{=}3{,}100$-dimensional $\ker(W)$ achieves only $r{=}0.208$.
Normalized by dimensionality, the per-dimension phylogenetic signal density in $\mathrm{row}(W)$ is approximately 600 times that of $\ker(W)$, confirming that the learned projection successfully separates phylogenetic structure from acoustic texture in the VAE latent space.

\begin{figure}[htb!]
\centering
\includegraphics[width=\linewidth]{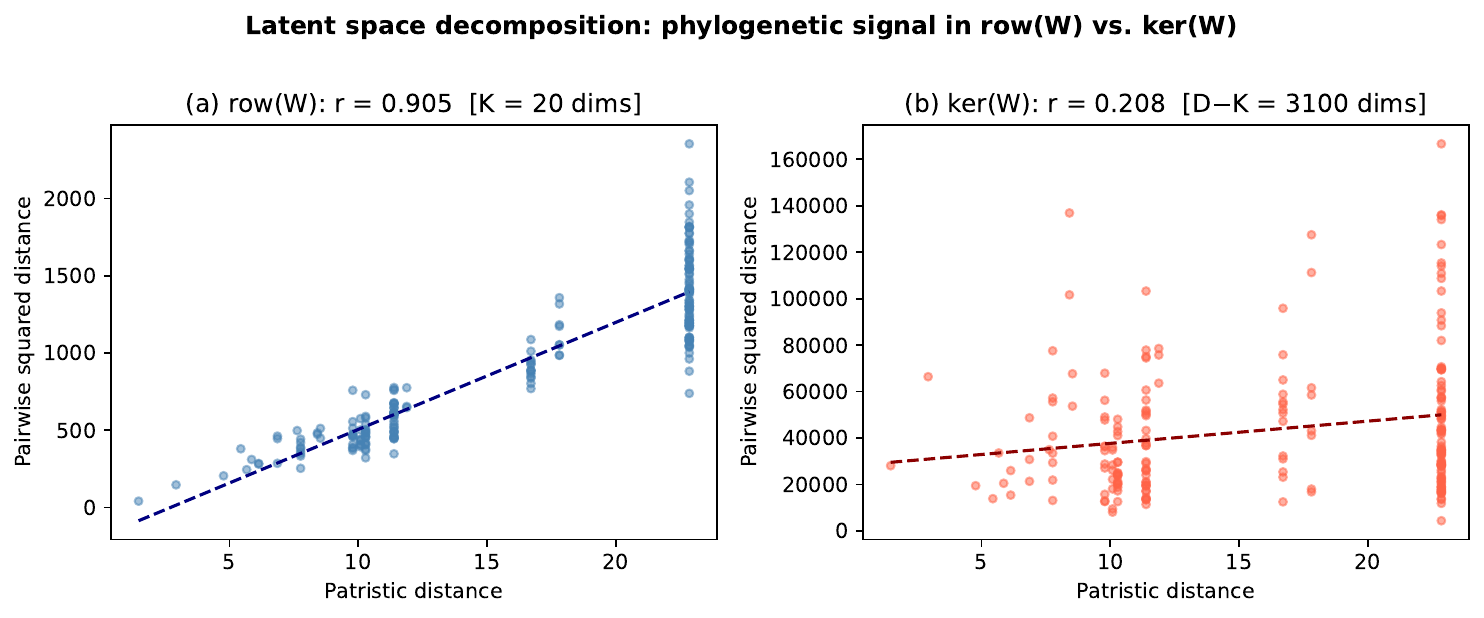}
\caption{\textbf{Subspace decomposition validation on Dataset~1.}
\textbf{(a)}~Pairwise squared distances in $\mathrm{row}(W)$ vs.\ patristic distances show a strong linear relationship at $r{=}0.905$, confirming that the learned trait subspace captures phylogenetic structure.
\textbf{(b)}~The same analysis in $\ker(W)$ yields only $r{=}0.208$ despite spanning 3{,}100 dimensions, indicating that the nullspace primarily encodes non-phylogenetic acoustic variation.}
\label{fig:subspace_decomp}
\end{figure}

\section{Residual Perturbation Properties}\label{app:resid_prop}

The pooled residual standard deviation is computed as $\hat{\sigma}_{\mathrm{res},k} = \bigl(\frac{1}{N}\sum_{s \in \mathcal{L}}\sum_{i=1}^{n_s}(z_{s,i,k} - y_{s,k})^2\bigr)^{1/2}$ with $N = \sum_s n_s$.

\begin{proposition}[Mean preservation and variance characterization]\label{prop:resid}
The perturbation~\eqref{eq:resid} satisfies:
\begin{enumerate}[nosep,leftmargin=15pt]
  \item $\E[\tilde{y}_u^{(b)} \mid y_u^{(b)}] = y_u^{(b)}$\quad (the BM posterior mean is preserved);
  \item $\mathrm{Cov}(\tilde{y}_u^{(b)} \mid y_u^{(b)}) = \gamma^2 \operatorname{diag}(\hat{\sigma}_{\mathrm{res}}^{\,2})$\quad (independent, axis-aligned noise);
  \item The marginal variance of $\tilde{y}_u^{(b)}$ integrating over both BM sampling and perturbation is $\tau^2 \sigma_u^2 \hat{\sigma}_{y,k}^2 + \gamma^2 \hat{\sigma}_{\mathrm{res},k}^2$ along coordinate $k$.
\end{enumerate}
\end{proposition}
\begin{proof}
Parts 1 and 2 follow from $\E[\epsilon] = 0$ and $\mathrm{Cov}(\epsilon) = I_K$. For part (3), by the law of total variance,
\[
\mathrm{Var}(\tilde{y}_{u,k}^{(b)}) = \E[\mathrm{Var}(\tilde{y}_{u,k}^{(b)} \mid y_u^{(b)})] + \mathrm{Var}(\E[\tilde{y}_{u,k}^{(b)} \mid y_u^{(b)}]) = \gamma^2 \hat{\sigma}_{\mathrm{res},k}^2 + \tau^2 \sigma_u^2 \hat{\sigma}_{y,k}^2. \qedhere
\]
\end{proof}

\section{Proofs and Theoretical Details}\label{app:lift_proofs}

\subsection{Proof of Proposition~\ref{prop:orth} (Exact Row-Orthonormality)}

By construction $LL^\top = AA^\top + \varepsilon I_K$, so $AA^\top = LL^\top - \varepsilon I_K$. Substituting into $WW^\top$:
\[
\begin{aligned}
WW^\top
  &= L^{-1}A A^\top L^{-\top}
   = L^{-1}(LL^\top - \varepsilon I_K)L^{-\top} \\
  &= L^{-1}LL^\top L^{-\top}
   - \varepsilon\,L^{-1}L^{-\top}
   = I_K - \varepsilon\,L^{-1}L^{-\top},
\end{aligned}
\]
where $L^{-1}L = I_K$ and $L^\top L^{-\top} = I_K$ give the first term. As $\varepsilon \to 0$, $WW^\top \to I_K$.

\smallskip\noindent\textbf{Operator norm bound.}\quad
For the quantitative bound, we relate $L^{-1}L^{-\top}$ to the spectrum of $AA^\top + \varepsilon I_K$.
Since $L$ is square and invertible, $L^{-1}L^{-\top} = (L^\top L)^{-1}$.
For any invertible square matrix $M$, the matrices $M^\top M$ and $MM^\top$ share the same eigenvalues (both equal the squared singular values of $M$), so
\[
\|L^{-1}L^{-\top}\|_{\mathrm{op}} = \|(L^\top L)^{-1}\|_{\mathrm{op}} = \|(LL^\top)^{-1}\|_{\mathrm{op}} = \frac{1}{\lambda_{\min}(LL^\top)} = \frac{1}{\lambda_{\min}(AA^\top + \varepsilon I_K)}.
\]
Since the eigenvalues of $AA^\top$ are the squared singular values of $A$, we have $\lambda_{\min}(AA^\top + \varepsilon I_K) = \sigma_{\min}^2(A) + \varepsilon$, yielding
\[
\|WW^\top - I_K\|_{\mathrm{op}} = \varepsilon\,\|L^{-1}L^{-\top}\|_{\mathrm{op}} = \frac{\varepsilon}{\sigma_{\min}^2(A) + \varepsilon} \le \frac{\varepsilon}{\sigma_{\min}^2(A)},
\]
where the last step uses $\sigma_{\min}^2(A) + \varepsilon > \sigma_{\min}^2(A)$. This yields the bound stated in Proposition~\ref{prop:orth}. In practice we use $\varepsilon = 10^{-6}$ and observe $\sigma_{\min}(A) \approx 0.1$, giving a deviation of approximately $10^{-4}$, which is negligible compared to the scale of BM posterior uncertainty. This also bounds the trait-preservation error of the anchored lift (Lemma~\ref{lem:lift}): $\|Wx^\star - y\|_2 \le \|WW^\top - I_K\|_{\mathrm{op}} \cdot \|y - Wx_0\|_2 \le 10^{-4}\|y - Wx_0\|_2$.

The Cholesky reparametrization has three advantages over direct Stiefel projection: (i)~it is differentiable everywhere (assuming full rank), enabling standard gradient-based optimization; (ii)~the Cholesky decomposition is numerically stable for $K \ll D$ (our setting: $K{=}20$, $D{=}3{,}120$); (iii)~it avoids the retraction step required by Riemannian optimizers on the Stiefel manifold, which can be expensive for tall-thin matrices.

\subsection{Proof of Proposition~\ref{prop:bm} (BM Conditional Posterior)}

The joint distribution of $[\bar{Y}_u; \bar{Y}_{\mathcal{L}}]$ under BM is Gaussian with block covariance
$\bigl[\begin{smallmatrix} C_{uu} & c_u^\top \\ c_u & \Sigma_{\mathcal{LL}} \end{smallmatrix}\bigr]$.
By the standard Gaussian conditioning formula, the conditional mean is $c_u^\top \Sigma_{\mathcal{LL}}^{-1} \bar{Y}_{\mathcal{L},j}$ and the conditional variance is $C_{uu} - c_u^\top \Sigma_{\mathcal{LL}}^{-1} c_u$. The nugget enters only through $\Sigma_{\mathcal{LL}}$ (it models observation noise at the tips), not through $c_u$ or $C_{uu}$. Coordinate-wise factorization follows from the i.i.d.-across-dimensions BM assumption.

In implementation, we solve the linear system $\Sigma_{\mathcal{LL}} \alpha = \bar{Y}_{\mathcal{L},j}$ via LU decomposition rather than explicitly computing $\Sigma_{\mathcal{LL}}^{-1}$, and clamp the posterior variance to $\max(\sigma_u^2, 0)$ to handle numerical errors near the root.

\subsection{Proof of Lemma~\ref{lem:lift} (Optimal Anchored Lift)}

Introduce Lagrange multipliers $\lambda \in \R^K$ and form the Lagrangian
\[
  \mathcal{L}(x, \lambda) = \tfrac{1}{2}\|x - x_0\|_2^2 + \lambda^\top\!\bigl(y - Wx\bigr).
\]
The first-order stationarity condition $\nabla_{x} \mathcal{L} = 0$ gives $x - x_0 = W^\top \lambda$, i.e.\ $x = x_0 + W^\top \lambda$. Substituting into the constraint $Wx = y$:
\[
  Wx_0 + WW^\top \lambda = y \implies \lambda = (WW^\top)^{-1}(y - Wx_0) = y - Wx_0,
\]
where the last equality uses $WW^\top = I_K$. Hence $x^\star = x_0 + W^\top(y - Wx_0)$. Uniqueness follows from strict convexity of $\|\cdot\|_2^2$ over the affine constraint set $\mathcal{A}(y)$.

\subsection{Proofs of Corollaries~\ref{cor:trait} and~\ref{cor:null}}

\noindent\textbf{Corollary~\ref{cor:trait}} (Exact trait preservation).
Direct computation: $Wx^\star = Wx_0 + WW^\top(y - Wx_0) = Wx_0 + I_K(y - Wx_0) = y$.

\noindent\textbf{Corollary~\ref{cor:null}} (Nullspace preservation).
From Lemma~\ref{lem:lift}, $x^\star - x_0 = W^\top v$ for $v = y - Wx_0 \in \R^K$. Since $W^\top v \in \mathrm{col}(W^\top) = \mathrm{row}(W) = \ker(W)^\perp$, we have $P_{\ker(W)}(x^\star - x_0) = 0$.

\subsection{Orthogonal Decomposition of the Latent Space}

The row-orthonormality $WW^\top = I_K$ induces the orthogonal decomposition $\R^D = \mathrm{row}(W) \oplus \ker(W)$, where $\dim(\mathrm{row}(W)) = K$ and $\dim(\ker(W)) = D - K$. The orthogonal projectors onto these subspaces are $P_{\mathrm{row}} = W^\top W$ and $P_{\ker} = I_D - W^\top W$, respectively. By Corollaries~\ref{cor:trait} and~\ref{cor:null}, the anchored lift operator acts as follows: (i)~it replaces the row-space component of $x_0$ with the target trait, the \emph{phylogenetic subspace} controlled by the BM posterior; (ii)~it preserves the nullspace component of $x_0$ exactly, the \emph{acoustic nullspace}, which encodes fine-grained spectro-temporal texture inherited from a nearby real clip.

This decomposition has a useful consequence for generation diversity. Let $x_0, x_0'$ be two different anchor clips with the same trait, i.e.\ $\phi_W(x_0) = \phi_W(x_0')$. Then the lifted latents $x^\star, x^{\star\prime}$ have the same row-space component but differ in their nullspace components: $P_{\ker}x^\star = P_{\ker}x_0 \ne P_{\ker}x_0' = P_{\ker}x^{\star\prime}$. This means different anchor selections produce distinct ancestral audio with the same phylogenetic content but different acoustic textures.

\subsection{Invariance of the Anchored Lift under Orthogonal Rotation}\label{app:rotation_inv}

As noted in Section~\ref{sec:trait}, $W$ is unique only up to left-multiplication by an orthogonal $Q \in \R^{K \times K}$. We verify that the full anchored lift, including anchor selection, is invariant to this ambiguity. Replacing $W$ with $\tilde{W} = QW$ gives traits $\tilde{y} = QWx = Qy$. Anchor selection~\eqref{eq:anchor} is invariant because $\|\tilde{W}x - \tilde{y}\|_2 = \|Q(Wx - y)\|_2 = \|Wx - y\|_2$ for any orthogonal $Q$, so the same anchor $x_0$ is selected under both $W$ and $\tilde{W}$. The lift then becomes
\[
  \tilde{x}^\star = x_0 + \tilde{W}^\top(\tilde{y} - \tilde{W}x_0) = x_0 + W^\top Q^\top Q(y - Wx_0) = x_0 + W^\top(y - Wx_0) = x^\star,
\]
where we used $Q^\top Q = I_K$. Hence the deterministic pipeline, comprising the posterior mean, anchored lift, and decoding, produces identical audio regardless of which orthogonal representative of $W$ is found by optimization. The stochastic sampling steps~\eqref{eq:bm-sample}--\eqref{eq:resid} are not strictly invariant because the axis-aligned noise scaling by $\hat{\sigma}_y$ and $\hat{\sigma}_{\mathrm{res}}$ does not commute with rotation, as $\mathrm{diag}(Q\Sigma Q^\top) \neq Q\,\mathrm{diag}(\Sigma)\,Q^\top$ in general. In practice this effect is negligible because the posterior mean dominates the sample, $\tau < 1$ attenuates phylogenetic noise, and $\gamma$ is small.

\section{Full Algorithm}\label{app:algorithm}

\begin{algorithm}[h]
\caption{Ancestral Sound Reconstruction via Tree-Metric Latent Traits}
\label{alg:tree_metric_asr}
\begin{algorithmic}[1]
\STATE \textbf{Input:} Tree $T = (\mathcal{V}, \mathcal{E})$; clip sets $\{\mathcal{A}_s\}_{s\in\mathcal{L}}$; pretrained $(E_{\mathrm{vae}}, D_{\mathrm{vae}}, V)$; hyperparameters $(K, \alpha, \lambda, B, \tau, \gamma)$.
\STATE \textbf{Output:} Ancestral waveforms $\{\hat{a}_u^{(b)}\}_{u\in\mathcal{U},\,b=1}^B$.

\STATE Encode clips: $x_{s,i} \leftarrow \mathrm{vec}(E_{\mathrm{vae}}(a_{s,i})) \in \R^D$ for all $s \in \mathcal{L},\, i = 1,\ldots,n_s$.
\STATE Learn $W \in \R^{K \times D}$ with $WW^\top = I_K$ via~\eqref{eq:trait-obj}--\eqref{eq:chol}.
\STATE Compute trimmed tip traits $\{y_s\}_{s \in \mathcal{L}}$ via~\eqref{eq:agg}; save empirical mean $\hat{\mu}_y$ and std.\ dev.\ $\hat{\sigma}_y$; standardize to zero mean, unit variance.
\STATE Build $\Sigma_{\mathcal{LL}} = C_{\mathcal{LL}} + \lambda I$; compute pooled within-species residual scale $\hat{\sigma}_{\mathrm{res}}$.
\FOR{each internal node $u \in \mathcal{U}$}
    \STATE Compute posterior mean $m_u$ and variance $\sigma_u^2$ via Proposition~\ref{prop:bm}.
    \FOR{$b = 1,\ldots, B$}
        \STATE Compute unstandardized mean $\hat{y}_u \leftarrow \hat{\mu}_y + \hat{\sigma}_y \odot m_u$; sample $y_u^{(b)} \leftarrow \hat{y}_u + \tau\,\sigma_u\,(\hat{\sigma}_y \odot \epsilon_1)$, $\epsilon_1 \sim \mathcal{N}(0, I_K)$.
        \STATE Perturb: $\tilde{y}_u^{(b)} \leftarrow y_u^{(b)} + \gamma\, \hat{\sigma}_{\mathrm{res}} \odot \epsilon_2$, \quad $\epsilon_2 \sim \mathcal{N}(0, I_K)$.
        \STATE Select anchor: $x_0 \leftarrow \arg\min_{x \in \mathcal{X}_{\mathrm{train}}} \|\phi_W(x) - \tilde{y}_u^{(b)}\|_2$.
        \STATE Lift: $x^\star \leftarrow x_0 + W^\top(\tilde{y}_u^{(b)} - Wx_0)$.
        \STATE Emit: $\hat{a}_u^{(b)} \leftarrow V\bigl(D_{\mathrm{vae}}(x^\star)\bigr)$.
    \ENDFOR
\ENDFOR
\end{algorithmic}
\end{algorithm}

Algorithm~\ref{alg:tree_metric_asr} summarizes the complete ancestral sound reconstruction pipeline described in Sections~\ref{sec:trait}--\ref{sec:lift}.

\section{Alternative Evolutionary Models}\label{app:evo_models}

In the main text we adopt standard Brownian motion (BM) as the evolutionary model for ancestral trait inference.
Here we describe two families of generalizations, the Ornstein--Uhlenbeck (OU) process and Pagel's branch-length transformations, and compare them empirically against BM.
All models share the same Gaussian conditioning machinery (Proposition~\ref{prop:bm}); they differ only in the tip-tip covariance $C_{\mathcal{LL}}$, the cross-covariance $c_u$, and the self-variance $C_{uu}$ used for each internal node $u$.

\subsection{Ornstein--Uhlenbeck Process}

The OU process~\citep{hansen1997stabilizing} adds a mean-reverting drift to Brownian motion.
Each trait coordinate evolves as
\begin{equation}\label{eq:ou-sde}
  dY_t = -\alpha\,(Y_t - \theta)\,dt + dW_t, \qquad \alpha > 0,
\end{equation}
where $\alpha$ controls the strength of stabilizing selection and $\theta$ is the long-run optimum, set to zero after standardization. The resulting covariance between tips $i,j$ separated by path lengths $s_i, s_j$ from their MRCA at depth $t_a$ is
\begin{equation}\label{eq:ou-cov}
  C_{ij}^{\mathrm{OU}} = \frac{1}{2\alpha}\bigl(1 - e^{-2\alpha\,t_a}\bigr)\,e^{-\alpha(s_i + s_j)},
\end{equation}
where $t_a = \mathrm{depth}(\mathrm{MRCA}(i,j))$ and $s_i, s_j$ are the branch lengths from the MRCA to tips $i, j$ respectively.
As $\alpha \to 0$ the OU covariance recovers BM; large $\alpha$ shrinks off-diagonal entries toward zero, effectively decorrelating distant relatives.

For the cross-covariance between internal node $u$, at depth $t_u$, and tip $i$, and the self-variance of $u$:
\begin{align}
  (c_u^{\mathrm{OU}})_i &= \frac{1}{2\alpha}\bigl(1 - e^{-2\alpha\,t_a}\bigr)\,e^{-\alpha\,(s_u + s_i)}, \label{eq:ou-cross}\\
  C_{uu}^{\mathrm{OU}} &= \frac{1}{2\alpha}\bigl(1 - e^{-2\alpha\,t_u}\bigr), \label{eq:ou-self}
\end{align}
where $t_a = \mathrm{depth}(\mathrm{MRCA}(u,i))$, $s_u = t_u - t_a$ is the branch length from the MRCA to node $u$, and $s_i$ is the branch length from the MRCA to tip $i$.
When $u$ is an ancestor of $i$, the MRCA is $u$ itself, so $s_u = 0$ and the expression reduces to $\frac{1}{2\alpha}(1 - e^{-2\alpha\,t_u})\,e^{-\alpha\,s_i}$.

\subsection{Pagel's \texorpdfstring{$\lambda$}{lambda} Transformation}

Pagel's $\lambda$~\citep{pagel1999inferring} rescales the off-diagonal elements of the BM covariance matrix to modulate phylogenetic signal strength:
\begin{equation}\label{eq:pagel-lam}
  C_{ij}^{\lambda} =
  \begin{cases}
    \lambda \cdot \mathrm{depth}\bigl(\mathrm{MRCA}(i,j)\bigr), & i \ne j,\\
    \mathrm{depth}(i), & i = j.
  \end{cases}
\end{equation}
When $\lambda = 1$ this recovers BM; $\lambda = 0$ produces a star phylogeny in which tips are independent given the root, discarding all shared evolutionary history.
Intermediate values interpolate between these extremes, allowing the model to down-weight phylogenetic covariance when the observed trait variation contains a substantial non-phylogenetic component.

The cross-covariance and self-variance for internal node $u$ are:
\begin{align}
  (c_u^{\lambda})_i &= \lambda \cdot \mathrm{depth}\bigl(\mathrm{MRCA}(u,i)\bigr), \label{eq:lam-cross}\\
  C_{uu}^{\lambda} &= \lambda \cdot \mathrm{depth}(u). \label{eq:lam-self}
\end{align}

\subsection{Pagel's \texorpdfstring{$\delta$}{delta} Transformation}

Pagel's $\delta$~\citep{pagel1999inferring} applies a power transformation to node depths, re-weighting early versus late evolution:
\begin{equation}\label{eq:pagel-del}
  C_{ij}^{\delta} = \bigl[\mathrm{depth}\bigl(\mathrm{MRCA}(i,j)\bigr)\bigr]^{\delta}.
\end{equation}
When $\delta = 1$ this recovers BM.
Values $\delta < 1$ stretch deep branches relative to shallow ones, allocating more variance to early divergence, consistent with early-burst models, while $\delta > 1$ compresses deep branches and stretches shallow ones, consistent with accelerating late evolution.

The cross-covariance and self-variance follow the same power transformation:
\begin{align}
  (c_u^{\delta})_i &= \bigl[\mathrm{depth}\bigl(\mathrm{MRCA}(u,i)\bigr)\bigr]^{\delta}, \label{eq:del-cross}\\
  C_{uu}^{\delta} &= \bigl[\mathrm{depth}(u)\bigr]^{\delta}. \label{eq:del-self}
\end{align}

\subsection{Empirical Comparison}

We evaluate 12 configurations on Dataset~1, Tyrannidae, 21 species, 20 internal nodes, spanning BM, OU with $\alpha \in \{0.1, 0.5, 1.0, 2.0\}$, Pagel's $\lambda$ with $\lambda \in \{0, 0.25, 0.5, 0.75\}$, and Pagel's $\delta$ with $\delta \in \{0.5, 1.5, 2.0\}$.
All configurations use the same learned tree-metric trait space, nearest-trait anchored lifting, and BigVGAN decoding pipeline as the full method in the main text; only the evolutionary model for ancestral inference differs.
Results are shown in Table~\ref{tab:evo_models}.

\begin{table}[htb!]
\caption{\textbf{Evolutionary model comparison.} All configurations use our learned tree-metric trait space with nearest-trait anchored lifting and mild residual perturbation ($\gamma=0.25$). Only the evolutionary model for ancestral inference varies. Best value per column in \textbf{bold}.}
\label{tab:evo_models}
\centering
\small
\renewcommand{\arraystretch}{1.06}
\resizebox{\linewidth}{!}{
\begin{tabular}{@{}ll|ccccc|ccccc@{}}
\toprule
& & \multicolumn{5}{c|}{\textit{Phylogenetic Plausibility}} & \multicolumn{5}{c}{\textit{Audio Quality}} \\
Model & Param & TreeCorr$\uparrow$ & DescMarg$\uparrow$ & EdgeMono$\uparrow$ & NNDist$\uparrow$ & Unique$\uparrow$ & CycleSim$\uparrow$ & FAD-BN$\downarrow$ & FAD-PANN$\downarrow$ & FAD-CLAP$\downarrow$ & Silence$\downarrow$ \\
\midrule
\rowcolor{cyan!5} BM & --- & 0.4859 & 0.0417 & 0.8667 & 0.2350 & 1.0000 & 0.9985 & 0.6326 & 0.1744 & 0.6710 & 0.0079 \\
\midrule
OU & $\alpha{=}0.1$ & 0.4983 & 0.0463 & \textbf{0.9333} & 0.2157 & 1.0000 & 0.9986 & 0.6097 & 0.1703 & 0.6627 & 0.0146 \\
OU & $\alpha{=}0.5$ & 0.1643 & 0.0044 & 0.8667 & 0.3916 & 1.0000 & 0.9981 & 0.9026 & 0.4534 & 1.3264 & 0.0261 \\
OU & $\alpha{=}1.0$ & 0.2443 & 0.0143 & 0.6667 & \textbf{0.4155} & 1.0000 & 0.9955 & 0.9332 & 0.6624 & 1.7026 & 0.0023 \\
OU & $\alpha{=}2.0$ & 0.1655 & 0.0088 & 0.8667 & 0.4102 & 1.0000 & 0.9911 & 0.9466 & 0.8766 & 1.8562 & \textbf{0.0002} \\
\midrule
Pagel $\lambda$ & $\lambda{=}0$ & 0.5226 & 0.0432 & 0.7333 & 0.2310 & 1.0000 & \textbf{0.9987} & 0.6454 & 0.1499 & 0.6332 & 0.0249 \\
Pagel $\lambda$ & $\lambda{=}0.25$ & \textbf{0.5730} & 0.0439 & 0.7333 & 0.2116 & 1.0000 & 0.9986 & \textbf{0.5729} & 0.1576 & \textbf{0.6339} & 0.0003 \\
Pagel $\lambda$ & $\lambda{=}0.5$ & 0.5150 & 0.0445 & 0.8000 & 0.2231 & 1.0000 & 0.9986 & 0.5990 & 0.1954 & 0.6872 & 0.0510 \\
Pagel $\lambda$ & $\lambda{=}0.75$ & 0.4705 & 0.0447 & 0.8000 & 0.2322 & 1.0000 & 0.9986 & 0.6294 & 0.1895 & 0.6903 & 0.0040 \\
\midrule
Pagel $\delta$ & $\delta{=}0.5$ & 0.3877 & \textbf{0.0812} & 0.6667 & 0.2088 & 1.0000 & 0.9986 & 0.5687 & \textbf{0.1433} & 0.6593 & 0.1133 \\
Pagel $\delta$ & $\delta{=}1.5$ & 0.4674 & 0.0160 & 0.6000 & 0.2014 & 1.0000 & 0.9985 & 0.6455 & 0.2156 & 0.6978 & 0.0159 \\
Pagel $\delta$ & $\delta{=}2.0$ & 0.4374 & 0.0158 & 0.6000 & 0.2044 & 1.0000 & 0.9985 & 0.6521 & 0.2154 & 0.6992 & 0.0154 \\
\bottomrule
\end{tabular}
}
\end{table}

\noindent\textbf{Discussion.}\quad
Pagel's $\lambda{=}0.25$ achieves the highest Tree--Embedding Correlation, 0.5730 vs.\ 0.4859 for BM, an 18\% improvement, and competitive audio quality with FAD-BN of 0.5729, indicating better global phylogenetic consistency.
However, BM retains stronger Edge Monotonicity, 0.8667 vs.\ 0.7333, suggesting that the BM posterior better preserves local parent--child ordering.
OU with mild selection at $\alpha{=}0.1$ performs comparably to BM with a slight edge in Tree Correlation and Edge Monotonicity, but stronger OU selection at $\alpha \geq 0.5$ degrades all metrics substantially as the mean-reverting drift overwhelms phylogenetic signal.
Pagel's $\delta{=}0.5$ achieves the best Descendant Affinity Margin of 0.0812, the lowest FAD-PANN of 0.1433, and the lowest FAD-BN of 0.5687, but at the cost of higher Silence fraction of 0.1133 and weaker Edge Monotonicity.

Overall, BM provides the best balance across all metrics and remains our default choice in the main text.
Pagel's $\lambda$ with mild phylogenetic signal reduction at $\lambda \approx 0.25$ is a promising alternative when global tree correlation is prioritized over local edge ordering.

\section{Experimental Details}\label{app:exp_details}

\begin{figure}[tb!]
\centering
\includegraphics[width=0.95\linewidth]{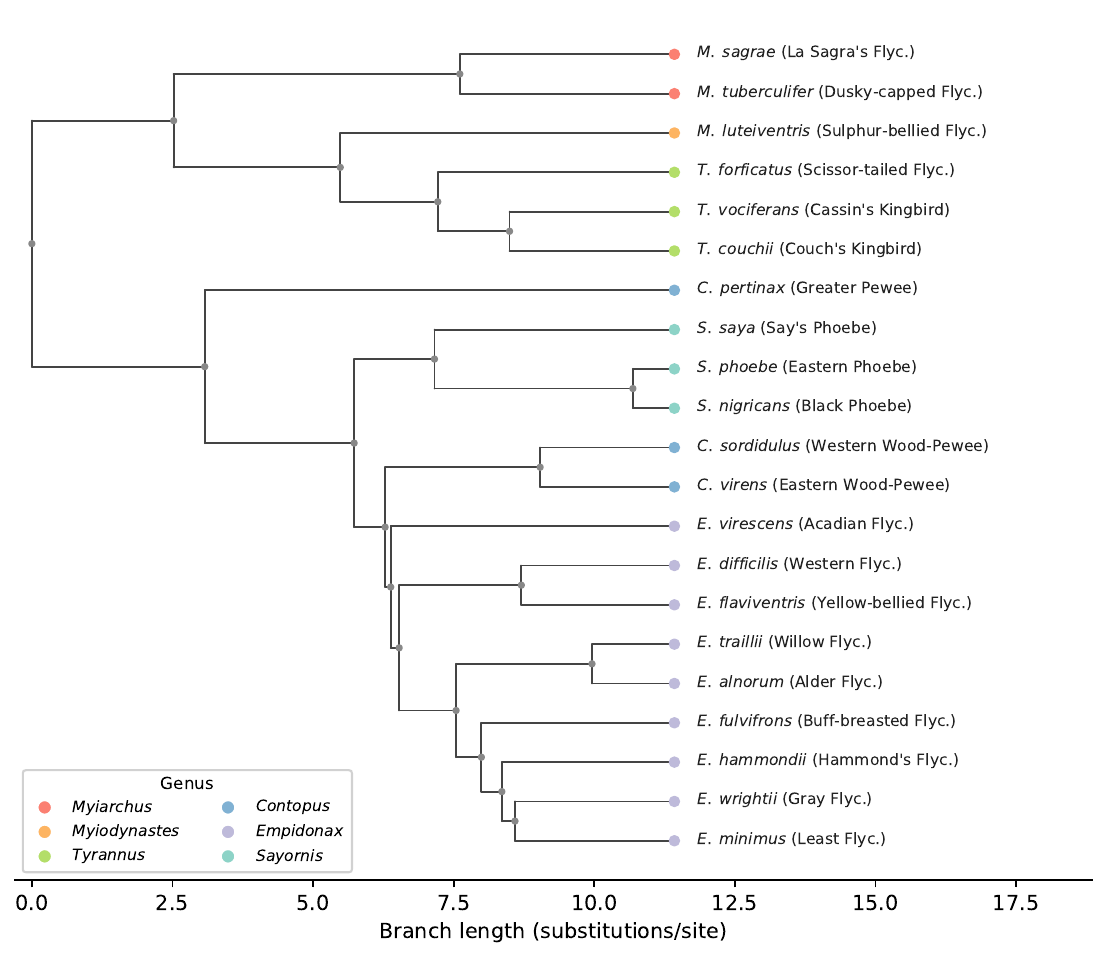}
\caption{\textbf{Phylogenetic tree of the tyrant flycatcher clade for Dataset~1.}
21 extant species from 6 genera, with 20 internal ancestral nodes, derived from BirdTree.org~\citep{jetz2012global}.
Branch lengths reflect evolutionary distance in substitutions per site.
Species are colored by genus.}
\label{fig:phylo_tree}
\end{figure}

\subsection{Dataset 1: Tyrannidae}\label{app:dataset1}

Our first dataset consists of vocalization recordings from a clade of New World tyrant flycatchers, family Tyrannidae, one of the most species-rich passerine radiations.
The 21 focal species were selected to maximize phylogenetic coverage across 6 genera, \textit{Empidonax}, \textit{Contopus}, \textit{Myiarchus}, \textit{Myiodynastes}, \textit{Sayornis}, and \textit{Tyrannus}, and to include only species with sufficient high-quality recordings available on Xeno-Canto with quality ratings A--B and a minimum duration of 5\,s.
The clade comprises 10 \textit{Empidonax} spp., 3 \textit{Contopus} spp., 3 \textit{Sayornis} spp., 2 \textit{Tyrannus} spp., 2 \textit{Myiarchus} spp., and 1 \textit{Myiodynastes} sp.
All audio files were manually reviewed and trimmed to standardized 5-second clips using Raven Pro v1.6, then converted to 16-bit signed 16\,kHz mono WAV format.
Each species is represented by 47--68 audio clips with a mean of 54.7 clips per species, for a total of approximately 1{,}100 clips.
For the phylogenetic framework, 10{,}000 posterior trees were downloaded from BirdTree.org~\citep{jetz2012global} using the Ericson Stage~2 backbone~\citep{ericson2006diversification}.
Species name concordance between the focal species list and the BirdTree taxonomy was verified before pruning.
All 10{,}000 trees were pruned to the focal species set using the R packages \texttt{ape}~\citep{paradis2004ape} and \texttt{phangorn}~\citep{schliep2011phangorn}, and a maximum clade credibility consensus tree was computed by selecting the tree with the highest product of posterior clade support values while retaining the original branch lengths.
The resulting tree is fully resolved with no polytomies, with 21 tip nodes and 20 internal ancestral nodes and is ultrametric after midpoint rooting, with a uniform root-to-tip depth of 11.42 substitutions/site.
Figure~\ref{fig:phylo_tree} shows the tree topology with species colored by genus.
The tree structure captures several well-resolved sister-species pairs, such as \textit{E.\ traillii}--\textit{E.\ alnorum} and \textit{S.\ phoebe}--\textit{S.\ saya}, that serve as natural test cases for ancestral reconstruction: the most recent common ancestor of a closely related pair should produce a vocalization intermediate between its two extant descendants.

\begin{figure}[tb!]
\centering
\includegraphics[width=0.95\linewidth]{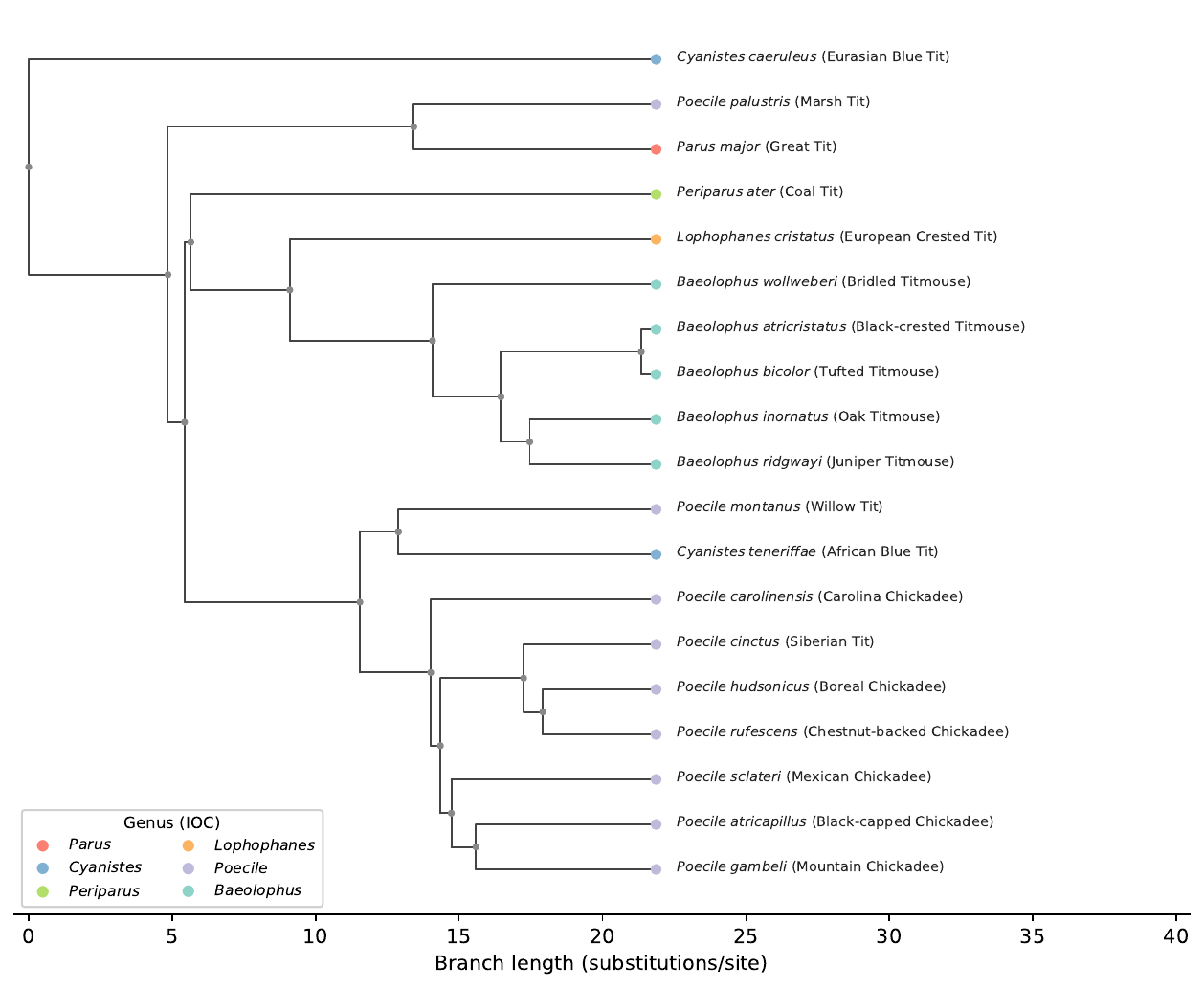}
\caption{\textbf{Phylogenetic tree of the Paridae clade for Dataset~2.}
19 extant species from 6 genera under IOC taxonomy, with 18 internal ancestral nodes, derived from BirdTree.org~\citep{jetz2012global}.
Branch lengths reflect evolutionary distance in substitutions per site.
Species are colored by IOC genus; the BirdTree taxonomy lumps \textit{Cyanistes}, \textit{Periparus}, \textit{Lophophanes}, and \textit{Poecile} under \textit{Parus}.}
\label{fig:phylo_tree_xc}
\end{figure}

\subsection{Dataset 2: Paridae}\label{app:dataset2}

Our second dataset tests cross-clade generalization using an independently curated set of tits, chickadees, and titmice, family Paridae.
The 19 focal species span 6 genera: 9 \textit{Poecile} spp., 5 \textit{Baeolophus} spp., 2 \textit{Cyanistes} spp., 1 \textit{Parus} sp., 1 \textit{Periparus} sp., and 1 \textit{Lophophanes} sp.
Recordings were downloaded programmatically from the Xeno-Canto API with quality ratings A--B, up to 80 per species.
Unlike Dataset~1, which was manually curated with Raven Pro, recordings shorter than 3\,s were excluded automatically and the loudest 5-second segment was extracted from each remaining file; recordings between 3--5\,s were zero-padded to the target length.
After filtering, 50--80 clips per species were retained, totaling approximately 1{,}150 clips.
Audio is standardized identically to Dataset~1: 16-bit signed 16\,kHz mono WAV, 5 seconds per clip.
For the phylogenetic framework, 1{,}000 posterior trees were downloaded from BirdTree.org~\citep{jetz2012global} using the Hackett All Species backbone~\citep{hackett2008phylogenomic}, which differs from the Ericson Stage~2 backbone used for Dataset~1.
The BirdTree taxonomy lumps \textit{Cyanistes}, \textit{Periparus}, \textit{Lophophanes}, and \textit{Poecile} under \textit{Parus}, so we maintain an explicit name mapping between IOC names used by Xeno-Canto and BirdTree tip names.
Posterior trees were pruned to the focal 19-species set and a consensus tree was computed using median branch lengths, implemented in Python with BioPython~\citep{cock2009biopython} rather than the R-based workflow used for Dataset~1.
The resulting tree has 19 tip nodes and 18 internal ancestral nodes.
The Paridae clade differs from Tyrannidae in several respects: it includes both Old and New World species rather than being restricted to a single biogeographic region, its vocalizations tend to be simpler call notes and whistled songs rather than complex flycatcher vocalizations, and its phylogenetic structure is shallower with less divergent branch lengths.
These differences, together with the distinct phylogenetic backbone and data curation pipeline, make it a stringent test of whether our method generalizes beyond the characteristics of a single dataset.

\subsection{Pretrained Components}\label{app:pretrained}

Our pipeline relies on two pretrained components that are frozen throughout all downstream experiments: a variational autoencoder and a neural vocoder. Neither component is fine-tuned or retrained for different datasets.

\smallskip\noindent\textbf{VAE.}\quad
We use the pretrained 1D convolutional VAE from Make-An-Audio~2~\citep{huang2023makeanaaudio2} without any retraining or fine-tuning on bird audio.
The model operates on 80-bin log-mel spectrograms and compresses each 5-second clip into a latent representation of shape $20 \times 156$, which we flatten to $D = 3{,}120$.
The same frozen weights are used for both datasets.

\smallskip\noindent\textbf{Vocoder.}\quad
Generated mel spectrograms are converted to 16\,kHz waveforms using a pretrained BigVGAN vocoder~\citep{lee2023bigvgan}.
BigVGAN is a high-fidelity neural vocoder based on large-receptive-field anti-aliased periodic activations, pretrained on general audio.
We use the vocoder as a frozen black box: the VAE decoder produces an 80-bin mel spectrogram, and BigVGAN synthesizes the corresponding waveform without any fine-tuning on bird audio.

\subsection{Tree-Metric Trait Learning}

The linear trait projection $W \in \R^{K \times D}$ is trained to maximize the Pearson correlation between pairwise squared trait distances and patristic distances (Eq.~\ref{eq:trait-obj}).
We optimize using Adam with learning rate $10^{-2}$ for 1{,}500 gradient steps.
The free parameter $A \in \R^{K \times D}$ is initialized from $\mathcal{N}(0, 0.05^2)$ and the Cholesky regularizer is $\varepsilon = 10^{-6}$.
We sweep over trait dimensions $K \in \{2, 4, 8, 16, 20, 21\}$ and select $K = 20$ based on downstream evaluation performance (Appendix~\ref{app:k_sensitivity}, Table~\ref{tab:k_sensitivity}).
Species-level trait representatives are formed by trimmed mean aggregation with trim fraction $\alpha = 0.3$.

\subsection{Ancestral Inference and Generation}

Brownian ancestral inference uses nugget regularization $\lambda = 10^{-3}$.
For each internal node, we draw $B = 3$ posterior samples with sampling temperature $\tau = 0.6$ and residual perturbation scale $\gamma = 0.25$.
Inverse lifting uses nearest-trait anchoring: for each ancestral trait sample, the anchor is selected as the training clip whose projected trait is nearest in $\ell_2$ distance.
Generated latents are decoded through the pretrained VAE decoder, and the resulting mel spectrograms are converted to 16\,kHz waveforms using a pretrained BigVGAN vocoder.

\subsection{Baseline Configurations}

All baselines use the same phylogenetic tree, the same trimmed-mean aggregation procedure with $\alpha=0.3$ applied in each method's own embedding space, and BM posterior sampling with nugget $\lambda=10^{-3}$, three posterior samples per node, and temperature $\tau=0.6$.
PhyloNN and Poincar\'e both use 20-dimensional embeddings, matching our trait dimension $K{=}20$, so that differences in performance reflect the representation rather than capacity.

\smallskip\noindent\textbf{Raw Latent BM} performs BM inference directly in the full 3{,}120-dimensional VAE latent space without any dimensionality reduction. Ancestral latents are decoded through the VAE decoder and vocoder.

\smallskip\noindent\textbf{BirdNET BM} extracts 1{,}024-dimensional embeddings from a pretrained BirdNET classifier~\citep{kahl2021birdnet} and uses these as the trait space for BM inference. Internal-node outputs are rendered via nearest-exemplar lookup: for each inferred ancestor embedding, the training clip with the closest BirdNET embedding in L2 distance is selected and its VAE latent is decoded.

\smallskip\noindent\textbf{Classical Acoustic} extracts a 33-dimensional feature vector per clip using \texttt{librosa}~\citep{mcfee2015librosa}, comprising 20 MFCCs, spectral centroid, spectral bandwidth, spectral rolloff, spectral flatness, zero-crossing rate, RMS energy, and 7-band spectral contrast. BM inference is performed on species-mean acoustic features, and rendering uses nearest-exemplar lookup in acoustic feature space.

\smallskip\noindent\textbf{PhyloNN}~\citep{elhamod2023phylonn} maps VAE latents to a 20-dimensional embedding via a two-layer MLP encoder with hidden dimension 512 and ReLU activation. The network is supervised with a species-level cross-entropy head and additional hierarchical clade-classification heads at three tree-depth levels, weighted by $\beta=0.5$. Training uses Adam with learning rate $10^{-3}$ for 800 epochs. BM inference is performed in the learned 20-dimensional space, and rendering uses nearest-exemplar lookup.

\smallskip\noindent\textbf{Poincar\'e}~\citep{carvallo2025siamese} uses the same MLP architecture as PhyloNN but projects outputs onto the Poincar\'e ball. The encoder is trained with an all-pairs contrastive loss with margin 2.0: same-species pairs are pulled together and different-species pairs are pushed apart in Poincar\'e distance. Training uses Adam with learning rate $10^{-3}$ for 800 epochs. At inference, BM is performed in the tangent space at the origin via the logarithmic map, and rendering uses nearest-exemplar lookup.

\smallskip\noindent\textbf{Rendering pipeline for all retrieval baselines.}\quad
For every retrieval baseline, nearest-exemplar lookup is performed in the method's own embedding space. The selected training clip is then rendered by decoding its precomputed VAE latent through the same VAE decoder and BigVGAN vocoder used by generative methods, so all outputs undergo identical VAE encode--decode distortion and audio quality comparisons are fair.

\subsection{Evaluation Setup}

All evaluation metrics are computed on the generated ancestral audio for the internal nodes, 20 for Dataset~1 and 18 for Dataset~2.
For each internal node, the $B{=}3$ posterior samples are independently decoded to audio, embedded via BirdNET, and averaged to produce a single per-node embedding used in all phylogenetic plausibility metrics. FAD metrics pool all $B \times |\mathcal{U}|$ generated clips as the generated distribution.
BirdNET~\citep{kahl2021birdnet} 1{,}024-dimensional embeddings serve as the primary evaluation representation for phylogenetic plausibility metrics.
FAD metrics use three complementary embedding spaces: BirdNET at 1{,}024 dimensions, PANNs CNN14~\citep{kong2020panns} at 2{,}048 dimensions, and CLAP~\citep{elizalde2023clap} with HTSAT-tiny backbone at 512 dimensions.

\subsection{Computational Resources}\label{app:compute}

All experiments are conducted on a single NVIDIA L4 GPU (24\,GB).
The computational cost of our method is modest: the only learned component is the linear trait projection $W \in \R^{K \times D}$, which is optimized for 1{,}500 Adam steps and converges in under 10 seconds per dataset.
Baseline training (PhyloNN, Poincar\'e) requires 800 epochs each, completing in approximately 2--3 minutes per model.
The full ancestral reconstruction pipeline for one dataset, comprising latent encoding, BM posterior inference, anchored inverse lifting, VAE decoding, and BigVGAN vocoding for all internal nodes with $B{=}3$ samples each, runs in approximately 60 seconds.
Evaluation is dominated by embedding extraction: BirdNET inference runs on CPU via TFLite, while PANNs and CLAP each require a single forward pass on GPU, with all evaluation metrics for one dataset completing in under 5 minutes.
The total wall-clock time for a complete run, training all methods, generating ancestral audio, and computing all metrics for both datasets, is under 30 minutes on a single L4 GPU.

\section{Evaluation Metric Definitions}\label{app:eval_metrics}

We define the five phylogenetic plausibility metrics and five audio quality metrics used in the main text. All metrics are computed on the generated ancestral audio for internal nodes $u \in \mathcal{U}$, using BirdNET embeddings as the evaluation representation unless otherwise noted.

\subsection{Phylogenetic Plausibility Metrics}

All BirdNET embeddings are L2-normalized to unit length.
Let $e_u \in \R^{D_{\mathrm{bn}}}$ denote the BirdNET embedding of the generated audio for internal node $u$, and let $e_s$ denote the BirdNET centroid embedding of real clips for extant species $s \in \mathcal{L}$.
Define the cosine distance $d_{\cos}(a,b) = 1 - a^\top b$ for unit-normed vectors $a, b$.

\smallskip\noindent\textbf{Tree--Embedding Correlation.}\quad
We compute the Spearman rank correlation between pairwise phylogenetic distances and pairwise cosine distances over all pairs of nodes $i, j \in \mathcal{U} \cup \mathcal{L}$, where internal nodes use their generated-audio embeddings and tips use their real-clip centroid embeddings:
\begin{equation}
  \mathrm{TreeCorr} = \rho_{\mathrm{Spearman}}\!\bigl(\{d_T(i,j)\}_{i<j},\; \{d_{\cos}(e_i, e_j)\}_{i<j}\bigr).
\end{equation}
Higher values indicate that the generated audio preserves the global geometric structure of the phylogeny. Including tip centroids tests whether generated ancestors are placed correctly relative to real species.

\smallskip\noindent\textbf{Descendant Affinity Margin.}\quad
For each internal node $u$, let $\mathcal{D}_u \subset \mathcal{L}$ be its set of descendant tips and $\mathcal{N}_u = \mathcal{L} \setminus \mathcal{D}_u$ its non-descendants. We compute:
\begin{equation}
  \mathrm{DescMarg} = \frac{1}{|\mathcal{U}|}\sum_{u \in \mathcal{U}} \Bigl[\bar{s}(u, \mathcal{D}_u) - \bar{s}(u, \mathcal{N}_u)\Bigr],
\end{equation}
where $\bar{s}(u, \mathcal{S}) = |\mathcal{S}|^{-1}\sum_{s \in \mathcal{S}} e_u^\top e_s$ is the mean cosine similarity. Positive values indicate that ancestors are, on average, more similar in embedding space to their own descendants than to unrelated species.

\smallskip\noindent\textbf{Edge Monotonicity.}\quad
At each binary internal node $p$ with two children $c_1, c_2$ connected by branch lengths $b_1, b_2$, we check whether the child on the longer branch is also more distant from the parent in cosine space:
\begin{equation}
  \mathrm{EdgeMono} = \frac{1}{|\mathcal{B}|}\sum_{p \in \mathcal{B}} \mathbf{1}\!\bigl[(b_1 > b_2) \Leftrightarrow \bigl(d_{\cos}(e_p, e_{c_1}) > d_{\cos}(e_p, e_{c_2})\bigr)\bigr],
\end{equation}
where $\mathcal{B}$ is the set of binary internal nodes with non-tied branch lengths, and children $c_1, c_2$ may be either internal nodes or tips. This measures whether branch lengths predict local acoustic divergence: a child separated by a longer branch should be more acoustically distant from the parent.

\smallskip\noindent\textbf{Nearest-Train Distance.}\quad
For each internal node $u$, we compute the minimum cosine distance to any training clip in BirdNET embedding space:
\begin{equation}
  \mathrm{NNDist} = \frac{1}{|\mathcal{U}|}\sum_{u \in \mathcal{U}} \min_{x \in \mathcal{X}_{\mathrm{train}}} d_{\cos}(e_u, e_x),
\end{equation}
where $\mathcal{X}_{\mathrm{train}}$ denotes individual training clip embeddings rather than species centroids. Higher values indicate that the generated ancestral audio is not simply copying a nearby training exemplar, suggesting genuine interpolation rather than memorization.

\smallskip\noindent\textbf{Unique Rate.}\quad
We measure the fraction of ancestors that receive distinct audio outputs by comparing MD5 hashes of the generated waveform files:
\begin{equation}
  \mathrm{Unique} = \frac{|\{\text{distinct hashes among } \hat{a}_u,\; u \in \mathcal{U}\}|}{|\mathcal{U}|}.
\end{equation}
A Unique Rate below 1.0 indicates retrieval collapse, where multiple ancestors map to the same training exemplar and produce byte-identical waveforms. All generative methods achieve Unique Rate 1.0 by construction since they decode distinct latent vectors.

\subsection{Audio Quality Metrics}

\smallskip\noindent\textbf{Cycle Similarity (CycleSim).}\quad
We measure VAE round-trip fidelity at the mel-spectrogram level: for each generated ancestral audio, we compute the mel spectrogram $M$, encode and decode it through the VAE to obtain the reconstructed mel $\hat{M} = D_{\mathrm{vae}}(E_{\mathrm{vae}}(M))$, and report the cosine similarity between the flattened originals and reconstructions:
\begin{equation}
  \mathrm{CycleSim} = \frac{1}{|\mathcal{U}|}\sum_{u \in \mathcal{U}} \frac{\langle \mathrm{vec}(M_u),\; \mathrm{vec}(\hat{M}_u)\rangle}{\|\mathrm{vec}(M_u)\| \cdot \|\mathrm{vec}(\hat{M}_u)\|}.
\end{equation}
Values close to 1 indicate that the generated audio lies on the VAE's decodable manifold and survives re-encoding without degradation.

\smallskip\noindent\textbf{Fr\'echet Audio Distance (FAD).}\quad
We compute the Fr\'echet distance~\citep{kilgour2019frechet} between the distribution of generated ancestral audio embeddings and the distribution of real training clip embeddings, in three complementary embedding spaces:
\begin{itemize}[nosep,leftmargin=13pt]
  \item \textbf{FAD-BN}: BirdNET~\citep{kahl2021birdnet} 1{,}024-dimensional embeddings, capturing species-level bioacoustic identity.
  \item \textbf{FAD-PANN}: PANNs CNN14~\citep{kong2020panns} 2{,}048-dimensional embeddings, capturing general audio event semantics.
  \item \textbf{FAD-CLAP}: CLAP~\citep{elizalde2023clap} 512-dimensional embeddings with HTSAT-tiny backbone, capturing audio--language aligned representations.
\end{itemize}
Each FAD is computed as $\mathrm{FAD} = \|\mu_g - \mu_r\|^2 + \mathrm{Tr}\!\bigl(\Sigma_g + \Sigma_r - 2(\Sigma_g \Sigma_r)^{1/2}\bigr)$, where $\mu_g, \Sigma_g$ and $\mu_r, \Sigma_r$ are the mean and covariance of the generated and real embedding distributions, respectively. Lower FAD indicates that the generated audio is more similar to real bird vocalizations in aggregate. Because the generated pool is small relative to the embedding dimensionality, the sample covariance $\Sigma_g$ is rank-deficient, so absolute FAD values carry upward bias. Cross-method comparisons remain valid because all methods use identical pool sizes and the same reference distribution of ${\sim}1{,}100$ real clips.

\smallskip\noindent\textbf{Silence Rate.}\quad
We compute the mean fraction of near-silent frames across all generated ancestral audio clips. A frame is considered silent if its root-mean-square energy falls below $-60$\,dBFS. High silence rates indicate degenerate generation where the method produces partially or fully silent waveforms.

\section{Full Qualitative Results}\label{app:qualitative}

Figures~\ref{fig:qualitative_appendix_1} and~\ref{fig:qualitative_appendix_2} show mel spectrograms for all 20 reconstructed ancestral nodes on Dataset~1 (Tyrannidae) across all six methods.
Figures~\ref{fig:qualitative_appendix_xc_1} and~\ref{fig:qualitative_appendix_xc_2} show the corresponding results for all 18 ancestral nodes on Dataset~2 (Paridae).
A central distinction visible throughout these figures is between methods that \emph{generate} novel ancestral audio and those that \emph{retrieve} existing training clips.
The retrieval-based baselines, Classical Acoustic, PhyloNN, and Poincar\'e, operate in non-decodable feature spaces and must render internal-node outputs by selecting the nearest training exemplar.
As a consequence, their spectrograms are visually sharp but frequently duplicated: multiple internal nodes map to the same training clip, and the outputs cannot represent genuinely intermediate ancestral sounds.
This is reflected quantitatively in Table~\ref{tab:main_phylo_results}, where these methods exhibit Unique Rates well below 1.0.

By contrast, our method generates novel spectrograms for every internal node with Unique Rate 1.0.
The anchored inverse lift ensures that each output inherits realistic acoustic texture from a nearby training clip via the nullspace component while its phylogenetic content is set by the BM posterior via the trait subspace.
Across the 20 nodes, our method consistently produces spectrograms with richer harmonic structure, more naturalistic temporal patterning, and fewer silence artifacts compared to the baselines.
The diversity of generated outputs, visible in the variation of spectral patterns across different ancestor nodes, demonstrates that the pipeline produces phylogenetically graded reconstructions rather than collapsing to a small set of prototypical sounds.

\begin{figure}[htb!]
\centering
\includegraphics[width=\linewidth]{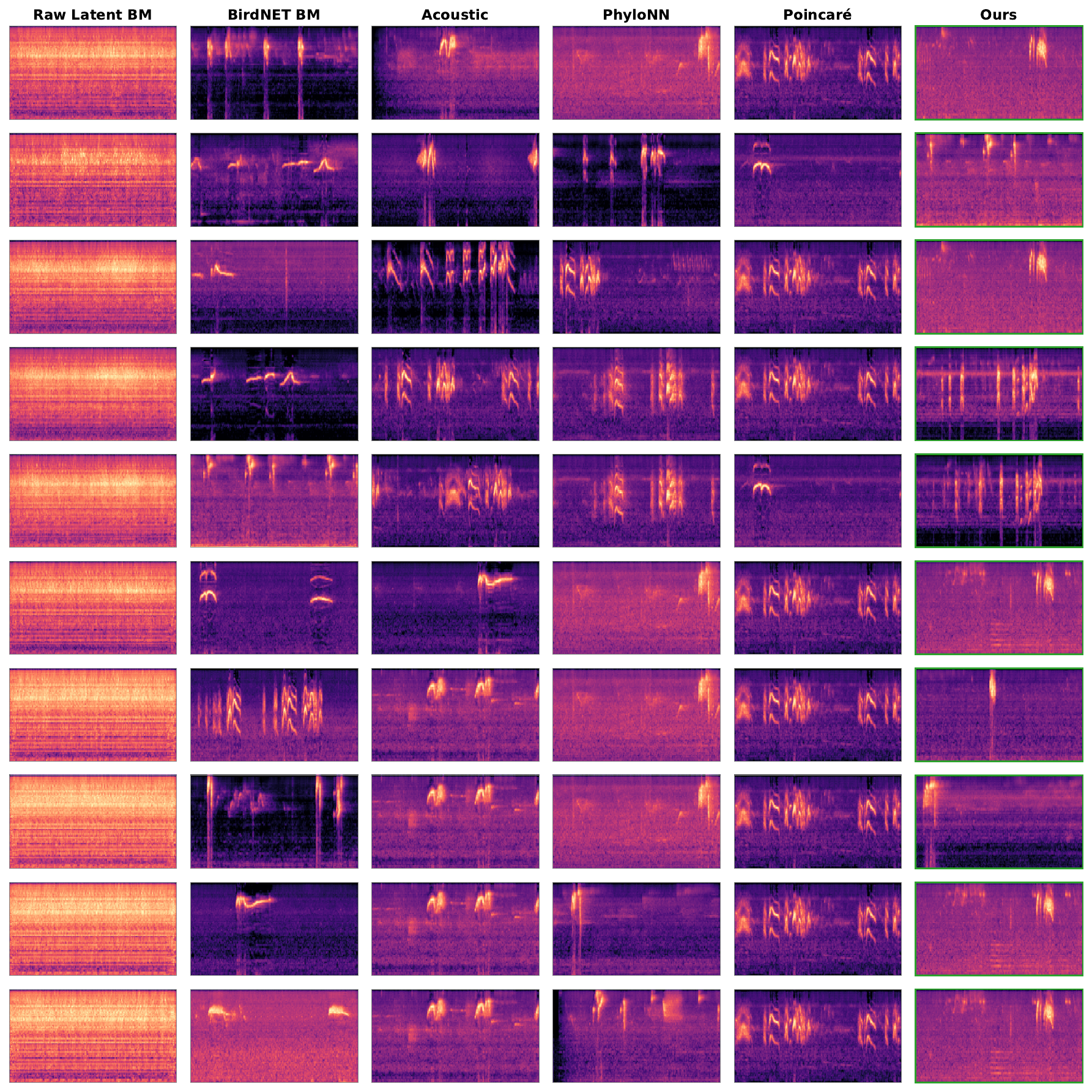}
\caption{\textbf{Ancestral sound reconstructions for internal nodes 1--10 on Dataset~1, Family Tyrannidae.}
Rows correspond to internal nodes labeled with ancestor index and genus; columns correspond to the six reconstruction methods.
Our method (green border) consistently produces spectrograms with richer harmonic structure and fewer silence artifacts compared to baselines.}
\label{fig:qualitative_appendix_1}
\end{figure}

\begin{figure}[htb!]
\centering
\includegraphics[width=\linewidth]{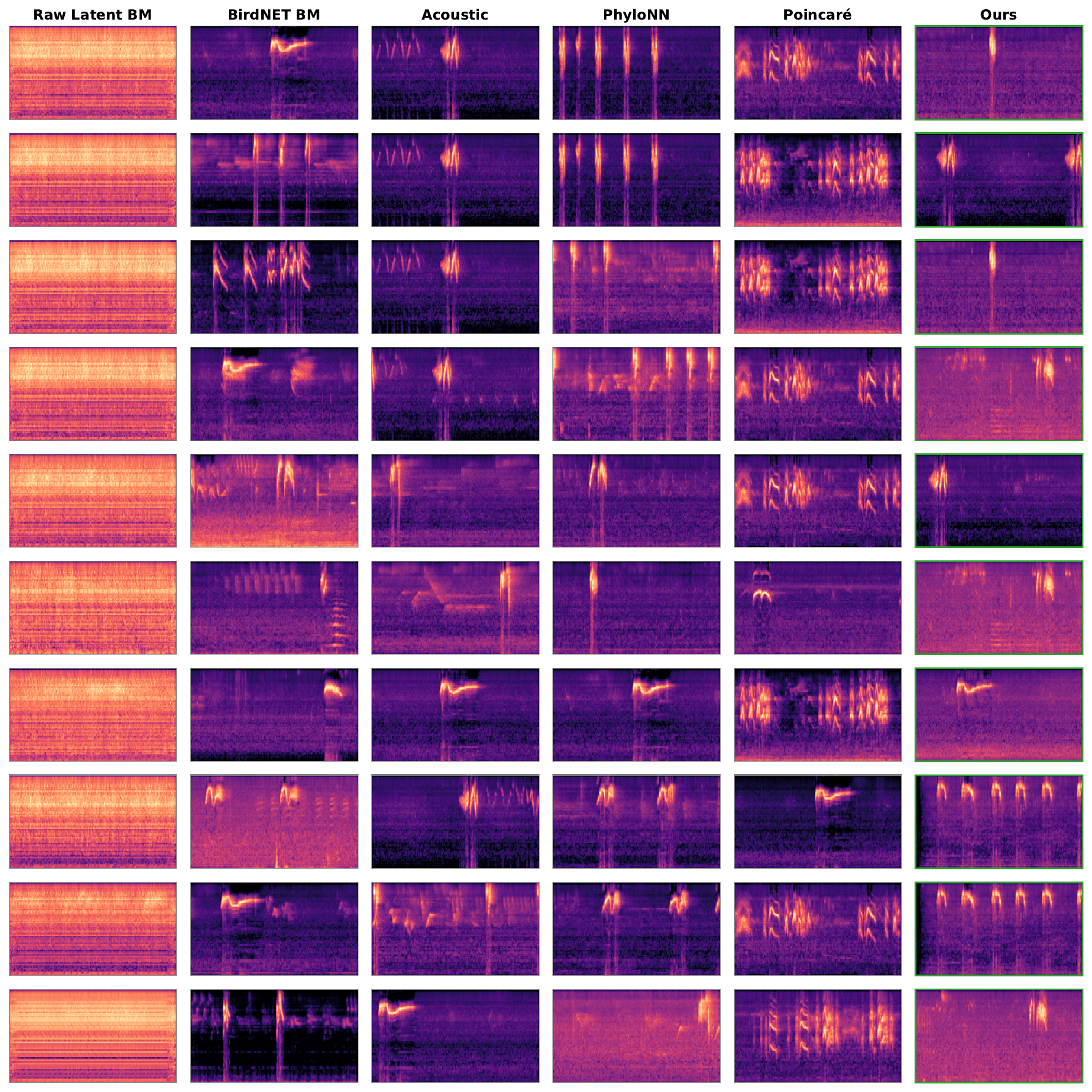}
\caption{\textbf{Ancestral sound reconstructions for internal nodes 11--20 on Dataset~1.}
Continuation of Figure~\ref{fig:qualitative_appendix_1}.}
\label{fig:qualitative_appendix_2}
\end{figure}

\begin{figure}[htb!]
\centering
\includegraphics[width=\linewidth]{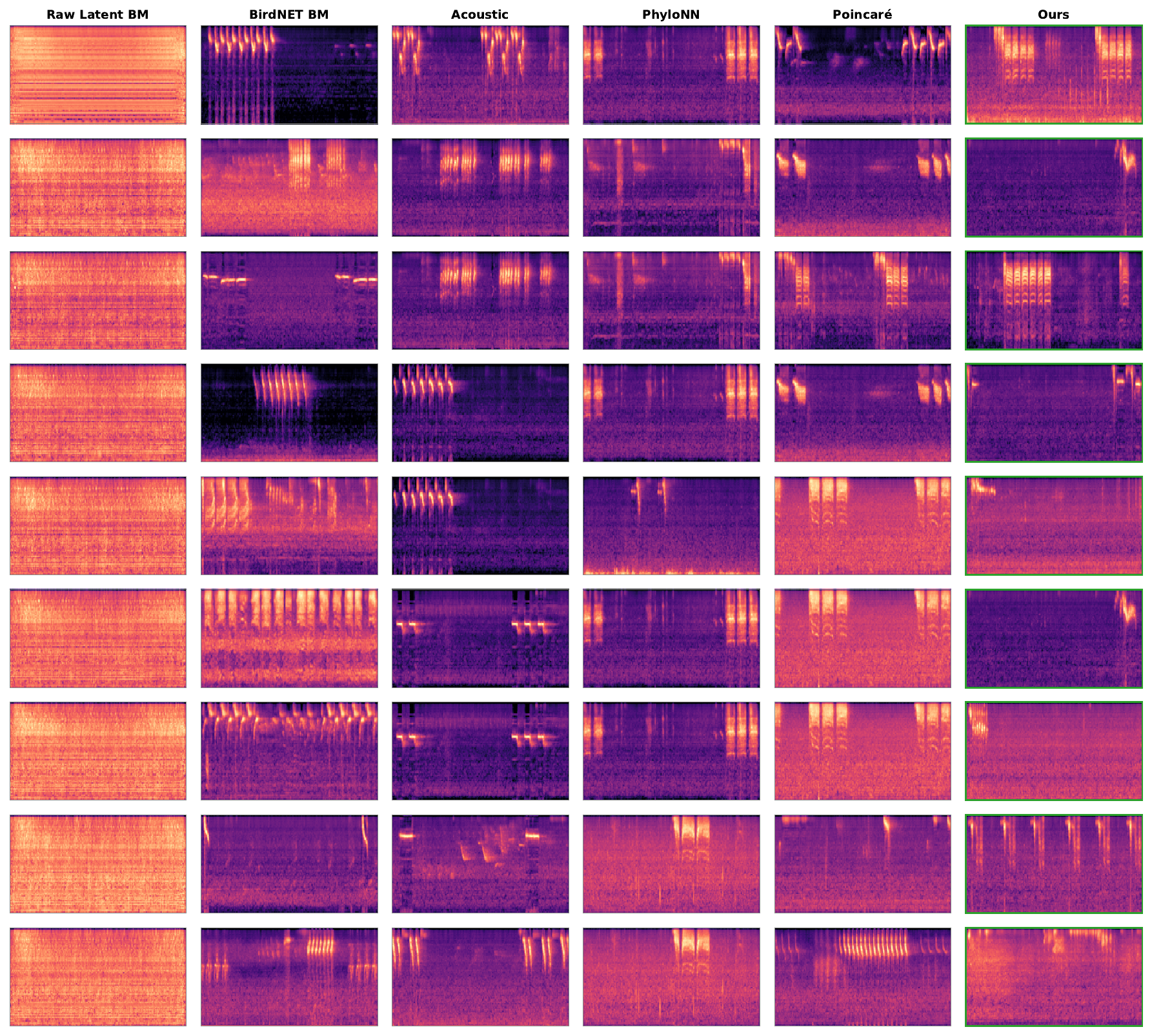}
\caption{\textbf{Ancestral sound reconstructions for internal nodes 1--9 on Dataset~2, Family Paridae.}
Same layout as Figures~\ref{fig:qualitative_appendix_1}--\ref{fig:qualitative_appendix_2} but on Dataset~2.
The same pattern holds: Raw Latent BM produces diffuse spectrograms, retrieval-based baselines duplicate training clips, and our method (green border) generates novel spectrograms with coherent harmonic structure.}
\label{fig:qualitative_appendix_xc_1}
\end{figure}

\begin{figure}[htb!]
\centering
\includegraphics[width=\linewidth]{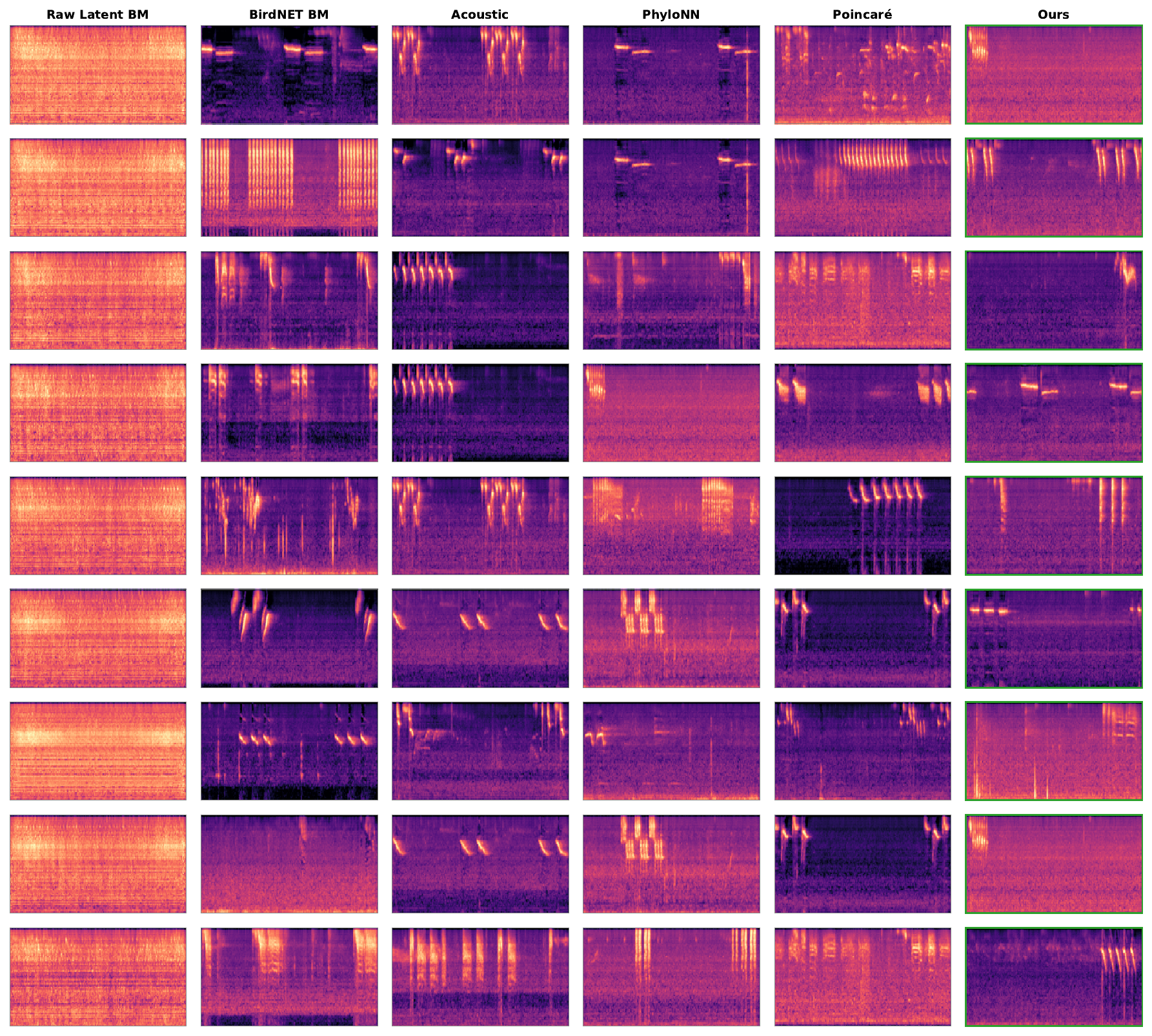}
\caption{\textbf{Ancestral sound reconstructions for internal nodes 10--18 on Dataset~2, Paridae.}
Continuation of Figure~\ref{fig:qualitative_appendix_xc_1}.}
\label{fig:qualitative_appendix_xc_2}
\end{figure}


\end{document}